\documentclass{article}
\usepackage{fullpage}
\usepackage{amsthm,amssymb,amsmath}  
\usepackage{authblk}
\usepackage{tikz} 

\usepackage{mathtools, eqparbox}%

\usepackage{wrapfig}
\usepackage[short,nocomma]{optidef}
\usepackage{cases}
\usepackage{todonotes}
\usepackage{subcaption}
\usepackage{nicefrac}
\usepackage{amsthm,amssymb,amsmath}  
\usepackage{xspace,enumerate}
\usepackage[utf8]{inputenc}
\usepackage{thmtools}
\usepackage{thm-restate}
\usepackage{todonotes}
\usepackage{hyperref}
\usepackage{verbatim}
\usepackage{multirow}
\usepackage{url}
\usepackage{mathtools}
\usepackage{amsthm}
\usepackage[utf8]{inputenc}
\usepackage{color}
\usepackage{caption}
\usepackage{scalerel}
\usepackage{subcaption}
\usepackage{todonotes}
\usepackage{algorithm}
\usepackage[noend]{algpseudocode}
\usepackage{algorithmicx}
\usepackage{etoolbox}
\makeatother
\usepackage[normalem]{ulem}
\usetikzlibrary{calc,snakes,shapes,arrows.meta}
\usepackage{booktabs}
\usepackage{bold-extra}

\newcommand{\cO}{\mathcal{O}}

\newcommand{\SSA}{\textsf{SSA}\xspace}
\newcommand{\SLCP}{\textsf{SLCP}\xspace}
\def\dd{\mathinner{.\,.}}

  \theoremstyle{plain}
  \newtheorem{theorem}{Theorem}
  \newtheorem{lemma}{Lemma}

  \theoremstyle{definition}
  \newtheorem{definition}{Definition}

  \newtheorem{remark}[definition]{Remark}
  {\bfseries}{\itshape}

\title{Sparse Suffix and LCP Array: Simple, Direct, Small, and Fast}
\author[1]{Lorraine A. K. Ayad}
\author[2]{Grigorios Loukides}
\author[3,4]{Solon P.\ Pissis}
\author[3]{Hilde Verbeek}

\affil[1]{Brunel University London, London, UK}
\affil[2]{King's College London, London, UK}
\affil[3]{CWI, Amsterdam, The Netherlands}
\affil[4]{Vrije Universiteit, Amsterdam, The Netherlands}

\date{}

\begin{document}

\maketitle

\begin{abstract}
Sparse suffix sorting is the problem of sorting $b=o(n)$ suffixes of a string of length $n$. Efficient sparse suffix sorting algorithms have existed for more than a decade. Despite the multitude of works and their justified claims for applications in text indexing, the existing algorithms have not been employed by practitioners. Arguably this is because there are no simple, direct, \emph{and} efficient algorithms for sparse suffix array construction. We provide two new algorithms for constructing the sparse suffix and LCP arrays that are simultaneously simple, direct, small, and fast. In particular, our algorithms are: \emph{simple} in the sense that they can be implemented using only basic data structures; \emph{direct} in the sense that the output arrays are not a byproduct of constructing the sparse suffix tree or an LCE data structure; \emph{fast} in the sense that they run in $\cO(n\log b)$ time, in the worst case, or in $\cO(n)$ time, when the total number of suffixes with an LCP value greater than $2^{\lfloor \log \frac{n}{b} \rfloor + 1}-1$ is in $\cO(b/\log b)$, matching the time of optimal yet much more complicated algorithms [Gawrychowski and Kociumaka, SODA 2017; Birenzwige et al., SODA 2020]; and \emph{small} in the sense that they can be implemented using \emph{only} $8b+o(b)$ machine words. Our algorithms are non-trivial space-efficient adaptations of the Monte Carlo algorithm by I et al.~for constructing the sparse suffix tree in $\cO(n\log b)$ time [STACS 2014]. We provide extensive experiments to justify our claims on simplicity and on efficiency.
\end{abstract}

\section{Introduction}

Let $T=T[1\dd n]$ be a string of length $n$ over an ordered alphabet $\Sigma$. Further let $\mathcal{B} \subseteq [1,n]$ be a set of $b>1$ positions in $T$. \emph{Sparse suffix sorting} is the problem of sorting the set of suffixes $T_{\mathcal{B}} = \{T[i \dd n] : i \in \mathcal{B}\}$ lexicographically~\cite{DBLP:conf/cocoon/KarkkainenU96}. This is achieved by constructing the sparse suffix array. The \emph{sparse suffix array} $\SSA=\SSA[1\dd b]$ is the array containing the positions in $\mathcal{B}$ in the lexicographical order of the suffixes in $T_{\mathcal{B}}$. 
The associated \emph{sparse longest common prefix array} $\SLCP=\SLCP[1\dd b]$ stores the length $\SLCP[i]$ of the longest common prefix of $T[\SSA[i-1]\dd n]$ and $T[\SSA[i]\dd n]$ when $i\in[2,n]$ or 0 when $i=1$. The $\SSA$ and $\SLCP$ array can be used to construct the sparse suffix tree in linear time using the algorithm by Kasai et al.~\cite{DBLP:conf/cpm/KasaiLAAP01}.
The \emph{sparse suffix tree} is the compacted trie of the set $T_{\mathcal{B}}$. Vice-versa, the \SSA and \SLCP array can be obtained in linear time using a pre-order traversal of the sparse suffix tree. 

Sparse suffix sorting was introduced as a fundamental step in the construction of compressed or sparse text indexes~\cite{DBLP:conf/cocoon/KarkkainenU96}. Modern compressed text indexes~\cite{DBLP:journals/tcs/NavarroP19,DBLP:journals/talg/ChristiansenEKN21}, practical indexes for long patterns~\cite{DBLP:journals/spe/GrabowskiR17,DBLP:conf/esa/LoukidesP21,DBLP:journals/tkde/LoukidesPS23,DBLP:journals/pvldb/AyadLP23}, and sublinear-space string algorithms~\cite{DBLP:conf/cpm/Nun0KK20,DBLP:conf/isaac/0001FGP23} rely on sparse suffix sorting: they first sample a sublinear number of ``important'' suffixes, which they next sort to construct their final solution. Efficient sparse suffix sorting algorithms have existed for more than a decade. The following algorithms construct $\SSA$ explicitly, or implicitly by first constructing the sparse suffix tree. 
Since the size of $\SSA$ (and the size of sparse suffix tree) is $\Theta(b)$, the goal of these algorithms is  
to use $\cO(b)$ words of space assuming read-only
random access to $T$:

\begin{itemize}
    \item K{\"{a}}rkk{\"{a}}inen et al.~presented a deterministic $\cO(n^2/s)$-time and $\cO(s)$-space algorithm, for any $s\in [b,n]$~\cite[Section 8]{DBLP:journals/jacm/KarkkainenSB06}.
    \item Bille et al.~presented a Monte Carlo $\cO(n\log^2 b)$-time and $\cO(b)$-space algorithm~\cite{DBLP:journals/talg/Bille0GKSV16}. They also presented a Las Vegas $\cO(n\log ^2n+b^2\log b)$-time and $\cO(b)$-space algorithm.
    \item I et al.~presented a Monte Carlo $\cO(n + (bn/s) \log s)$-time and $\cO(b)$-space algorithm, for any $s\in[b,n]$~\cite{DBLP:conf/stacs/IKK14}. They also presented a Las Vegas $\cO(n\log b)$-time and $\cO(b)$-space algorithm.
    \item Gawrychowski and Kociumaka~\cite{DBLP:conf/soda/GawrychowskiK17} presented a Monte Carlo $\cO(n)$-time and $\cO(b)$-space algorithm.
    They also presented a Las Vegas $\cO(n\sqrt{\log b})$-time and $\cO(b)$-space algorithm.
    \item Birenzwige et al.~\cite{DBLP:conf/soda/BirenzwigeGP20} presented a Las Vegas algorithm running in $\cO(n)$ time using $\cO(b)$ space.
    They also presented a deterministic $\cO(n\log \frac{n}{b})$-time and $\cO(b)$-space algorithm, for any $b=\Omega(\log n)$.
\end{itemize}   

The following algorithms also construct $\SSA$, but they work in the \emph{restore model}~\cite{DBLP:journals/talg/ChanMR18}: an algorithm is allowed to overwrite parts of the input, as long as it can restore it to its original form at termination.

\begin{itemize}   
    \item Fischer et al.~\cite{DBLP:journals/talg/FischerIK20} presented a deterministic $\cO(c\sqrt{\log n} + b \log b \log n \log^{*} n)$-time and $\cO(b)$-space algorithm, where $c$ is the number of letters that must be compared for distinguishing the suffixes in $T_{\mathcal{B}}$. 
    In some cases, this runs in sublinear extra time; extra refers to the linear cost of loading $T$ in memory.
    \item Prezza~\cite{DBLP:journals/talg/Prezza21} presented a Monte Carlo $\cO(n + b\log^2 n)$-time algorithm using $\cO(1)$ words of space. 
\end{itemize}

\paragraph{Motivation.} 
Despite the multitude of works on sparse suffix sorting and their justified claims for applications in text indexing, the existing algorithms have not been employed by practitioners. 
Arguably this is because there are no simple, direct, \emph{and} efficient algorithms for \SSA construction. 
The $\cO(n)$-time algorithms of Gawrychowski and Kociumaka~\cite{DBLP:conf/soda/GawrychowskiK17} and of Birenzwige et al.~\cite{DBLP:conf/soda/BirenzwigeGP20} are far from simple 
and do not seem to be practically promising either.
The former (Monte Carlo) algorithm relies heavily on the construction of compacted tries, which induce high constants in space usage, and on a recursive application of difference cover to construct a Longest Common Extension (LCE) data structure. The latter (Las Vegas) algorithm relies on an intricate partitioning scheme (sampling) to 
construct \SSA and on an LCE data structure to compute the \SLCP array. The Monte Carlo $\cO(n\log b)$-time algorithm of I et al.~\cite{DBLP:conf/stacs/IKK14} is simple but it also relies heavily on compacted tries, which makes it less likely to be employed by practitioners for \SSA construction.
The Monte Carlo $\cO(n + b\log^2 n)$-time algorithm of Prezza~\cite{DBLP:journals/talg/Prezza21} makes heavy usage of an LCE data structure as well: 
constructing the \SSA and \SLCP array is a byproduct of an in-place LCE data structure.
This algorithm is, to the best of our knowledge, the only algorithm 
which has been implemented (at least in a simplified form). 
Due to the interest in sparse suffix sorting and the above  characteristics of the existing algorithms, 
we were motivated to revisit this problem 
to develop efficient, yet simple and direct, algorithms for \SSA construction.
Such algorithms may serve as baselines for practitioners to engineer the \SSA and \SLCP array construction.

\paragraph{Our Results.}

We assume the standard word RAM model with word size $\Theta(\log n)$; basic arithmetic and bit-wise operations on $\cO(\log n)$-bit integers take $\cO(1)$ time. 
We assume that we have a read-only random access string $T$ of length $n$ over an integer alphabet $\Sigma=\{1,\ldots,n^{\cO(1)}\}$, a read-only integer array $A$ of size $b$ storing the $b$ elements of $\mathcal{B}$, and two write-only integer arrays \SSA and \SLCP, each of size $b$. We count the amount of \emph{extra space in machine words} used to construct the \SSA and \SLCP array. We present two algorithms:
\begin{enumerate}  
    \item Our first algorithm, \textsc{Main-Algo}, constructs \SSA and \SLCP \emph{directly}; i.e., without first explicitly constructing the sparse suffix tree or an LCE data structure. Its time complexity is $\cO(n+(bn/s)\log s))$ and its space complexity is $s+7b+o(b)$ machine words, for any chosen $s\in[b,n]$. It is a Monte Carlo algorithm that returns the correct output \emph{with high probability}; i.e., with probability at least $1 - n^{-c}$, for any constant $c\geq 1$ chosen at construction time. \textsc{Main-Algo} is \emph{simple} in the sense that it can be implemented using only \emph{basic data structures} (e.g., dictionaries and arrays) readily available in widely-used programming language (e.g., \texttt{C++}, \texttt{Java}, or \texttt{Python}). \textsc{Main-Algo} is a non-trivial space-efficient simulation of the algorithm by I et al.~for sparse suffix tree construction~\cite{DBLP:conf/stacs/IKK14}.
    A disadvantage of these two algorithms is that they attain the $\Theta(n \log b)$ time bound for $s=b$ \emph{in any case}. 
    To address this, we develop \textsc{Parameterized-Algo}, a parameterized algorithm which is input-sensitive.
    
    \item Our second algorithm, \textsc{Parameterized-Algo}, also constructs \SSA and \SLCP directly. Its time complexity is $\cO(n + (b'n/b) \log b)$ and its space complexity is $8b + 4b' + o(b)$ machine words, where $b'$ is the total number of suffixes $\SSA[i] \in \mathcal{B}$ with $\SLCP[i]\geq \ell$ or $\SLCP[i+1] \geq \ell$, where $\ell=2^{\lfloor \log \frac{n}{b} \rfloor + 1} - 1$. When $b'=\cO(b/\log b)$, \textsc{Parameterized-Algo} runs in $\cO(n)$ time, thus matching the time of the optimal yet much more complicated algorithms in~\cite{DBLP:conf/soda/GawrychowskiK17,DBLP:conf/soda/BirenzwigeGP20}, using \emph{only} $8b+o(b)$ machine words. It is a Monte Carlo algorithm that returns the correct output with high probability.
    It is \emph{remarkably simple} as it consists of two calls of \textsc{Main-Algo} and a linear-time step that merges the partial results (however, the proof of correctness requires some work). 
    The running time of \textsc{Parameterized-Algo} is good in the following sense: if the instance is reasonably sparse, then $\ell$ is large and likely $b'=\cO(b/\log b)$, thus it runs in $\cO(n)$ time. 
    In any case, it runs in $\cO(n \log b)$ time. For instance, for the full human genome (v.~GRCh38) as $T$, where $n\approx 3\cdot 10^9$, and for $b=\lfloor \sqrt{n} \rfloor=56137$ suffixes selected uniformly at random, $b'=2525< \lfloor  b/\log b \rfloor=3558$. We also analyze the time complexity of \textsc{Parameterized-Algo} on random strings and show that it works in $\cO(n)$ time (after a small amendment), for any string chosen uniformly at random from $\Sigma^n$ and any set $T_\mathcal{B}$ of $b$ suffixes of $T$, with high probability.
\end{enumerate}

To fully justify our claims, we also perform an experimental evaluation. For comparison purposes, we had to resort to Prezza's implementation of his algorithm from~\cite{DBLP:journals/talg/Prezza21}, which works in the restore model, and it is the only available implementation. Our experiments show that for realistically \emph{sparse instances} (i.e., for $b\in[n/10^7,n/10^3]$) of real-world datasets, our implementation of \textsc{Parameterized-Algo}: (i) runs in linear time (i.e., it is not substantially affected when $b$ increases); and (ii) clearly outperforms Prezza's implementation both in terms of time (even if the latter uses the $\Omega(n\log|\Sigma|)$ bits of $T$) \emph{and} space. For \emph{dense instances} (i.e., for $b\in[2 \cdot n/10^2,n/10]$), our implementation of \textsc{Parameterized-Algo} still runs faster than Prezza's implementation but, as expected, has a larger memory footprint.

The rest of the paper is organized as follows. Section~\ref{sec:prel} provides definitions and notation as well as a description of how Karp-Rabin fingerprints~\cite{DBLP:journals/ibmrd/KarpR87} are computed, stored, and used. In Section~\ref{sec:algorithm}, we present \textsc{Main-Algo}, as well as a full working example; and, in Section~\ref{sec:improved}, we present \textsc{Parameterized-Algo}. Finally, in Section~\ref{sec:experiments}, we present extensive experiments showing how \textsc{Main-Algo} and \textsc{Parameterized-Algo} perform in comparison to Prezza's implementation. 

\section{Preliminaries}\label{sec:prel}

We consider finite strings over an integer alphabet $\Sigma=\{1,\ldots,n^{\cO(1)}\}$.  The elements of $\Sigma$ are called \emph{letters}. A \emph{string} $T=T[1\dd n]$ is a sequence of letters from $\Sigma$; we denote by $|T|=n$ the \emph{length} of $T$. The fragment $T[i\dd j]$ of $T$ is an \emph{occurrence} of the underlying \emph{substring} $P=T[i]\ldots T[j]$ occurring at \emph{position} $i$ in $T$. A {\em prefix} of $T$ is a substring of $T$ of the form $T[1\dd j]$ 
and a {\em suffix} of $T$ is a substring of $T$ of the form $T[i\dd n]$. 

\paragraph{Karp-Rabin Fingerprints.}
Let $T$ be a string of length $n$ over an integer alphabet.
Let $p$ be a prime and choose $r \in[0,p-1]$ uniformly at random. 
The Karp-Rabin (KR) fingerprint~\cite{DBLP:journals/ibmrd/KarpR87} of $T[i\dd j]$ is defined as: 
$$\phi_T(i,j)=(\sum^{j}_{k=i}T[k]r^{j-k}\mod p, r^{j-i+1}\mod p).$$

Clearly, if $T[i\dd i + \ell] = T[j\dd j + \ell]$ then $\phi_T(i,i+\ell)=\phi_T(j,j+\ell)$. 
On the other hand, if
$T[i\dd i + \ell] \neq T[j\dd j + \ell]$ then $\phi_T(i,i+\ell)\neq \phi_T(j,j+\ell)$ with probability at least $1-\ell/p$~\cite{DBLP:conf/icalp/DietzfelbingerGMP92}.
Since we are comparing only substrings of equal length, the number of different possible
substring comparisons is less than $n^3$. Thus, for any constant $c\geq 1$, we can set $p$ to be a
prime larger than $\max(|\Sigma|,n^{c+3})$ to make the KR fingerprint
function perfect with probability at least $1 - n^{-c}$.
Any KR fingerprint of $T$ or $p$ fit in one machine word of size $\Theta(\log n)$.

\begin{remark}\label{rem:fixed-length}
In our algorithms (\textsc{Main-Algo} and \textsc{Parameterized-Algo}), \emph{the length} $j-i+1$ of the substrings of $T$ considered at each of the $\cO(\log n)$ iterations is fixed, hence the quantity $r^{j-i+1}\mod p$ only needs to be computed once before every iteration; and thus we do not store it together with the actual fingerprint to save machine words.
\end{remark}

\begin{lemma}[\cite{DBLP:conf/stacs/IKK14}]\label{lem:KRF}
Any string $T\in\Sigma^n$ can be preprocessed in $\cO(n)$ time using $s+\cO(1)$ machine words, for any $s\in[1,n]$, so that the KR fingerprint of any length-$k$ fragment of $T$ is computed in $\cO(\min\{k, n/s\})$ time.\footnote{I et al.~\cite{DBLP:conf/stacs/IKK14} claim $\cO(s)$ space but from their construction it is evident that in fact $s+\cO(1)$ machine words are used.}
\end{lemma}

I et al.~\cite{DBLP:conf/stacs/IKK14} employ the \emph{distribute-and-collect} technique~\cite{DBLP:journals/siamcomp/PaigeT87} to group $b$ suffixes, according to a fixed-length common prefix by using their KR fingerprints (see also Remark~\ref{rem:fixed-length}), in $\cO(b\log_s n)$ time. We instead use hashing to achieve the same result in $\cO(b)$ time with high probability. 
This gives improved running times in some special regimes (see Theorem~\ref{the:param} and Theorem~\ref{the:random}).
It also simplifies our algorithms in terms of implementation since efficient hashing implementations are readily available in any widely-used programming language.

\section{Main Algorithm}~\label{sec:algorithm}

\paragraph{Overview.} A summary of our 
main algorithm (\textsc{Main-Algo}) follows with references to the pseudocode given in Algorithms~\ref{alg:step1} and \ref{alg:step2}.\footnote{The pseudocode is complete in the sense that it only assumes the implementation of Lemma~\ref{lem:KRF} (Line 11).} As input, it takes a string $T$ from $\Sigma^n$ and an array $A$ of $b$ elements, indicating the starting positions of the suffixes that need to be sorted. As additional input, it takes an integer $j_\text{start}$, which defines the number of iterations. In this section,  $j_\text{start}$ is set to $\lfloor \log n \rfloor$, but a different value is used for the parameterized algorithm in Section~\ref{sec:improved}.

During the first phase (Algorithm~\ref{alg:step1}, Lines 1-25), the suffixes are distributed into groups such that all suffixes belonging to a particular group share a common prefix. At the end of this process, we are left with a hierarchy of groups that describes the exact longest common prefixes between suffixes. The members of each group are then sorted lexicographically (Algorithm~\ref{alg:step1}, Lines 26-27), which is made possible by knowing their longest common prefixes, such that a traversal of the hierarchy will yield the suffixes in lexicographic order. This traversal is the second phase (Algorithm~\ref{alg:step1}, Line 28 and Algorithm~\ref{alg:step2}): a simple depth-first search is used to construct the sparse suffix array and accompanying sparse LCP array from the hierarchy.

\subsection{Computing and Sorting the LCP Groups}
During the first phase, the suffixes of $A$ are organized into several LCP groups stored in set $B$. Each group in $B$ is represented by a triple $(i, k, \{v_1, \ldots, v_{n_i}\})$, where $i$ is the index (id) of the group, $k$ is its associated LCP value and $v_1$ through $v_{n_i}$ are its members, which are either suffixes or other groups. To distinguish between suffixes and groups, the indices $1$ through $b$ are reserved for the suffixes in $A$ and the group numbering starts at $b + 1$. At every point of the algorithm, it holds that in a group $(i, k, \{v_1, \ldots, v_{n_i}\})$, all suffixes and groups (with their respective suffixes) in $\{v_1, \ldots, v_{n_i}\}$ share a prefix of length \emph{at least} $k$. At the start (base case), there exists just one group $(b + 1, 0, \{1, \ldots, b\})$ (Line 3) containing all suffixes as members.

The LCP groups are then ``refined'' (the refinement process 
will be explained shortly) over the course of $\lfloor \log n \rfloor$ 
iterations, such that in the end each group describes the exact longest common prefix of its members rather than just a lower bound. Specifically, by the end of iteration $j$ (where $j$ descends from $\lfloor \log n \rfloor$ down to zero), in a group with LCP value $k$, two suffixes will have an actual longest common prefix of at least $k$ and at most $k + 2^j - 1$ letters. This gap is closed once $j$ has reached zero, at which point the refinement 
process is completed. The algorithm allows specifying a different starting value for $j$ than $\lfloor \log n \rfloor$, through the parameter $j_\text{start}$. This is used in the parameterized 
algorithm described in Section~\ref{sec:improved}.

The refinement process works as follows in iteration $j$. We refine every existing group; let one such group be $(i, k, \{v_1, \ldots, v_{n_i}\})$. We create a hash table (Line 8) and for every group member, with index $v_i$, we take the KR fingerprint as per Lemma~\ref{lem:KRF} of $T[A[v_i] + k \dd A[v_i] + k + 2^j - 1]$ (Line 11).\footnote{This assumes that $A[b_i] + k + 2^j - 1 \le n$; if this is not the case, the suffix ends at position $n$.} If $v_i$ denotes a group, we do the same thing using any given suffix belonging to that group; this is easily achieved by appending ``witness'' suffixes to $A$ for every created group as seen in Lines 4 and 22. (Any suffix can be a witness but we choose the one with the smallest index.) All the members are grouped in the hash table based on their KR fingerprints: if two suffixes have the same KR fingerprint, they will end up in the same entry of the hash table and with high probability have the same prefix of length $k + 2^j$. To save space, entries are removed from the group as they are added to the hash table (Line 13). All entries of the hash table are then inspected (Line 14). 
We distinguish three cases.
In case 1 (Lines 16-17), if all suffixes in a group end up having the same KR fingerprint, we update the LCP value of the old group to $k + 2^j$ rather than creating a new group.
In case 2 (Lines 18-22), if two or more suffixes have the same KR fingerprint, a new group is made with LCP value $k + 2^j$, containing these suffixes, and added to $B$ (Lines 18-22). 
After removing the suffixes from their original group, we replace them with the index of the newly created group (Line 20).
In case 3 (Lines 23-24), if a suffix is not grouped with any other suffix, we append it back to the group it originated from.

Once the iteration with $j = 0$ ends, all LCP groups describe the exact longest common prefix of their members.\footnote{This is generally not true when $j_\text{start}$ was set to a value less than $\lfloor \log n \rfloor$; in this case, the LCP values are only correct if they are at most $2^{j_\text{start} + 1} - 1$; this is expanded on in Section~\ref{sec:improved}.} We now sort the members of every group lexicographically (Lines 26-27). Sorting can be done using merge sort or radix sort, because these algorithms can be performed in place using $\cO(1)$ additional memory. Moreover, since we now know the exact LCP value for each group, two members in the same group can easily be compared in constant time: if they have a longest common prefix of length $k$, then the first position in which they differ is $k + 1$, meaning they can be compared by only comparing their $k+1$-th letters. After this, set $B$ contains the complete and sorted LCP groups, which are passed on to the second step of the algorithm.

\begin{algorithm}[H]\footnotesize
    \caption{\textsc{Main-Algo}}
    \label{alg:step1}
    \textbf{Input:} string $T \in \Sigma^n$, integer $b$, array $A$ of $b$ distinct integers from $[n]$, and integer $j_\text{start}$ (default $\lfloor \log n \rfloor$)  \\ 
    \textbf{Output:} $\SSA$ and $\SLCP$ 
    \begin{algorithmic}[1]
        \State{$m \gets b + 1$} \Comment{The number of currently used id's}
        \State{$L_m \gets (1,\dots,b)$}
        \Comment{The first group}
        \State{$B \gets \{(m, 0, L_m)\}$} \Comment{Set containing the first group $L_m$ with id $m$ and LCP value $0$} 
        \State{$A[m] \gets A[1]$} 
        \Comment{Choose a witness for group with id $m$ and extend $A$}
        \For{$j = j_\text{start},\ldots,0$}
        \Comment{For all relevant powers of $2$ down to $0$}
            \State{$B' \gets \emptyset$} \Comment{Set to store the new groups}
            \For{$(i, k, L_i) \in B$}
            \Comment{For every group}
                \State{$H_i \gets$ empty hash table}
                \Comment{key: the fingerprint of the group member; value: id's of the members}
                \State{$s \gets |L_i|$}
                \Comment{To memorize the original size of $L_i$}
                \For{$l \in L_i$}
                \Comment{Partition group $L_i$ according to LCP value $k + 2^j$}
                    \State{$h \gets \phi_T(A[l]+k, A[l]+k + 2^j-1)$} 
                    \Comment{Fingerprint of $T[A[l]\dd A[l] + k + 2^j-1]$ using Lemma~\ref{lem:KRF}}
                    \State{$H_i[h].\text{append}(l)$}
                    \Comment{Append id $l$ as a value of fingerprint $h$ to a satellite vector} 
                    \State{$L_i.\text{erase}(l)$}
                    \Comment{Remove $l$ from $L_i$ (to save space) -- $L_i$ is empty at the end of the loop}
                \EndFor
                \For{$h \in H_i$}
                \Comment{For every fingerprint in the hash table}
                    \State{$f \gets H_i[h]$}
                    \Comment{Let $f$ denote the vector of id's with fingerprint $h$}
                    \If{$|f| = s$}
                        \Comment{All members belong to the same group}
                        \State{$B \gets B \setminus \{(i, k, L_i)\} \cup  \{(i, k + 2^j, f)\}$}
                        \Comment{The old group is replaced due to a longer LCP}
                    \ElsIf{$|f| \ge 2$}
                    \Comment{New group $f$ must be created}
                        \State{$m \gets m + 1$}
                        \Comment{New group id}
                        \State{$L_i.\text{append}(m)$}
                        \Comment{Append the new id to $L$ of tuple $(i, k, L)\in B$}
                        \State{$B' \gets B' \cup \{(m, k + 2^j, f)\}$}
                        \Comment{New group $f$ with id $m$ and LCP value $k+2^j$ is created}
                        \State{$A[m] \gets A[f[1]]$} 
                        \Comment{Choose witness $A[f[1]]$ for group with id $m$ and extend $A$}
                    \ElsIf{$|f| =1$} 
                        \State{$L_i.\text{append}(f)$}
                        \Comment{Append the old id (only element of $f$) back to $L_i$}
                    \EndIf
                \EndFor
            \EndFor
           \State{$B\gets B \cup B'$}
           \Comment{Append $B'$ to $B$}
        \EndFor
        \For{$(i, k, L_i) \in B$}
        \State{$L_i.\text{sort}(l \mapsto T[A[l] + k])$}
        \Comment{Sort group in ascending order w.r.t. $T[A[l] + k]$ using in-place sort}
        \EndFor
        \State{\Return $\textsc{Output-Arrays}(B,b,A)$} \Comment{See Algorithm~\ref{alg:step2}}
    \end{algorithmic}
\end{algorithm}

\subsection{Constructing the \SSA and \SLCP Array}
The second phase (Algorithm~\ref{alg:step2}) of the main algorithm involves traversing the groups created in the previous phase in order to construct the \SSA and \SLCP array. At this point, the members of each group are sorted lexicographically, which means that the \SSA can be obtained by a simple pre-order walk along the hierarchy of the groups. For any two members, their exact longest common prefix is stored by their lowest common ancestor; that is, the group with the greatest LCP value that both suffixes fall under.

This part of the algorithm is thus a 
simple depth-first search of the underlying hierarchy that records all encountered suffixes in \SSA in the order they appear. 
For every group that is visited, the LCP value of its direct ``parent'' is stored with it (Lines 4 and 17). Throughout, a value $\ell$ is tracked that takes the value of the lowest LCP value that has been seen since the last suffix was encountered (Lines 8-9); every time a suffix is appended to \SSA, $\ell$ is appended to the \SLCP array (Lines 11-12). This completes the construction.

\begin{algorithm}[H]\footnotesize
    \caption{$\textsc{Output-Arrays}$}
    \label{alg:step2}
     \textbf{Input:} Set $B$ of tuples $(i, k, L_i)$ expected in ascending order by $i$, integer $b$, and array $A$\\ 
     \textbf{Output:} $\SSA$ and $\SLCP$
    \begin{algorithmic}[1]
        \State{$\SSA \gets$ empty array}
        \State{$\SLCP \gets$ empty array}
        \State{$S \gets$ empty stack}
        \State{$S.\text{push}((b + 1, 0))$} \Comment{The initial entry of the stack relating to the first group which has LCP value 0}
        \State{$\ell \gets 0$} \Comment{Variable tracking the least LCP value encountered since the last suffix}
        \While{$S$ is not empty}
            \State{$(i, \ell') \gets S.\text{pop}()$}
            \If{$\ell' < \ell$} \Comment{Update $\ell$ if it is less than before}
                \State{$\ell \gets \ell'$}
            \EndIf
            \If{$i \leq b$} \Comment{Member $i$ is a suffix: add it to $\SSA$ and $\SLCP$}
                \State{\SSA.append$(A[i])$}
                \State{\SLCP.append$(\ell)$}
                \State{$\ell \gets \infty$} \Comment{Reset $\ell$}
            \Else \Comment{Member $i$ is a group: add its members to the stack}
                \State{$(i, k, L_i) \gets B[i-b]$}
                \For{$i' \in L_i$ in reverse order} \Comment{Since the stack is FILO, the first entry needs to be added last}
                    \State{$S.\text{push}((i', k))$}
                \EndFor
            \EndIf
        \EndWhile
        \State{\Return \SSA and \SLCP}
    \end{algorithmic}
\end{algorithm}

\subsection{A Full Working Example}

A step-by-step example of the workings of this algorithm is described below. In this example, we construct the \SSA and \SLCP arrays on the string $T = \texttt{abracadabrarabia}$ for the positions $A = (1,3,8,10,11,13)$.
$A$ contains the starting indices of the suffixes we wish to process. The corresponding suffixes are, in order of their starting positions in $T$, 
\texttt{abracadabrarabia}, \texttt{racadabrarabia}, \texttt{abrarabia}, \texttt{rarabia}, \texttt{arabia} and \texttt{abia}. 

We start with a singular group containing all indices and having a shared prefix of length zero (see Table~\ref{tab:step1}). For decreasing powers of 2 starting with $2^{\lfloor \log n \rfloor}$, groups are refined in order to have more and more precise LCP values. For every new group that we create, we take the index corresponding to one of its members and add it to the array $A$, so that we can easily find the shared prefix corresponding to that group.

In the actual algorithm, the prefixes are efficiently grouped
by taking their KR fingerprints as per Lemma~\ref{lem:KRF} and indexed in a hash table. For the sake of the example, each relevant prefix is listed instead of its KR fingerprint. The members of each group are sorted lexicographically using merge sort or radix sort. At this point, the groups are refined to the point where they define the exact longest common prefix between their members, so sorting them is done by looking at just the  letter succeeding their longest common prefix.

\begin{table}[!t]
\begin{tabular}{|c|l l l l l|}
    \hline
    \multicolumn{6}{|c|}{$T = \texttt{abracadabrarabia}$ \;\;\;\; $A = (1,3,8,10,11,13)$} \\
    \multicolumn{6}{|c|}{Initial group: $(7, 0, (1, 2, 3, 4, 5, 6))$
    } \\
    \hline
    $2^j = 16, 8$ & \multicolumn{5}{l|}{None of the suffixes share a prefix of length $16$ or $8$, so no refinement is possible.} \\
    \hline
    \multirow{9}{*}{$2^j = 4$} & \multicolumn{5}{l|}{Consider the prefixes of length $4$ of every suffix:} \\
    & \multicolumn{5}{l|}{1: \texttt{\textbf{abra}}} \\
    & \multicolumn{5}{l|}{2: \texttt{raca}} \\
    & \multicolumn{5}{l|}{3: \textbf{\texttt{abra}}} \\
    & \multicolumn{5}{l|}{4: \texttt{rara}} \\
    & \multicolumn{5}{l|}{5: \texttt{arab}} \\
    & \multicolumn{5}{l|}{6: \texttt{abia}} \\
    & \multicolumn{5}{l|}{Combine the members 1 and 3 into a new group with id $8$ and LCP value 4:} \\
    & \multicolumn{5}{l|}{$\{(7, 0, (2, 4, 5, 6, \textbf{8})), (8, 4, (\textbf{1, 3}))\}$
    } \\
    \hline
    \multirow{8}{*}{$2^j = 2$} & \multicolumn{5}{l|}{Extend the prefixes by 2 and re-group:} \\
    & 2: \texttt{\textbf{ra}} & \multicolumn{4}{l|}{1: \texttt{(abra)ca}} \\
    & 4: \texttt{\textbf{ra}} & \multicolumn{4}{l|}{3: \texttt{(abra)ra}} \\
    & 5: \texttt{ar} & \multicolumn{4}{l|}{} \\
    & 6: \texttt{\textbf{ab}} & \multicolumn{4}{l|}{} \\
    & 8: \texttt{\textbf{ab}} & \multicolumn{4}{l|}{} \\
    & \multicolumn{5}{l|}{New groups:} \\
    & \multicolumn{5}{l|}{$\{(7, 0, (5, \textbf{9}, \textbf{10})), (8, 4, (1, 3)), (9, 2, (\textbf{2, 4})), (10, 2, (\textbf{6, 8}))\}$ 
    } \\
    \hline
    \multirow{6}{*}{$2^j = 1$} & \multicolumn{5}{l|}{Extend the prefixes of every suffix by 1 and regroup:} \\
    & 5: \texttt{\textbf{a}}  & 1: \texttt{(abra)c} & 2: \texttt{(ra)c} & \multicolumn{2}{l|}{6: \texttt{(ab)i}} \\
    & 9: \texttt{r}           & 3: \texttt{(abra)r} & 4: \texttt{(ra)r} & \multicolumn{2}{l|}{8: \texttt{(ab)r}} \\
    & 10: \texttt{\textbf{a}} & & & & \\
    & \multicolumn{5}{l|}{New groups:} \\
    & \multicolumn{5}{l|}{$\{(7, 0, (9, \textbf{11})), (8, 4, (1, 3)), (9, 2, (2, 4)), (10, 2, (6, 8)), (11, 1, (\textbf{5, 10}))\}$ 
    } \\
    \hline
    \multirow{5}{*}{Sorting} & \multicolumn{5}{l|}{Sort every group based on the letter appearing directly after the prefix of each member:} \\
    & 9: \texttt{r}  & 1: \texttt{(abra)c} & 2: \texttt{(ra)c} & 6: \texttt{(ab)i} & 5: \texttt{(a)r} \\
    & 11: \texttt{a} & 3: \texttt{(abra)r} & 4: \texttt{(ra)r} & 8: \texttt{(ab)r} & 10: \texttt{(a)b} \\
    & \multicolumn{5}{l|}{Final groups:} \\
    & \multicolumn{5}{l|}{$\{(7, 0, (\textbf{11, 9})), (8, 4, (1, 3)), (9, 2, (2, 4)), (10, 2, (6, 8)), (11, 1, (\textbf{10, 5}))\}$ 
    } \\
    \hline
\end{tabular}
\caption{Computing and sorting the LCP groups.}\label{tab:step1}
\end{table}

The second phase of the algorithm is a fairly straightforward DFS on the groups in order to build the two output arrays (see Table~\ref{tab:step2}). The variable $\ell$, at every point, takes as value the length of the longest common prefix stored in the lowest common ancestor between the last discovered suffix, in the (conceptual) underlying hierarchy that records all encountered suffixes in \SSA in the order they appear, 
and the currently inspected member. To facilitate this, every member is stored with its \emph{parent's} LCP value in the search stack.

At every iteration, the top of the stack (represented as the first member) is popped. If it is a group, its children are pushed onto the stack and $\ell$ is updated if necessary. If it is a suffix, the element $A[i]$, where $i$ is the index of the suffix, is added to \SSA, and the current value of $\ell$ is added to \SLCP. 

\begin{table}[!t]
\centering
\begin{tabular}{|c|c|l|l|l|}
    \hline
       & $\ell$ & \textbf{Stack} & \SSA & \SLCP \\
    \hline
     1 & $0$ & $(7, 0)$                         & $\emptyset$           & $\emptyset$ \\
     2 & $0$ & $(11, 0), (9, 0)$                & $\emptyset$           & $\emptyset$ \\
     3 & $0$ & $(10, 1), (5, 1), (9, 0)$        & $\emptyset$           & $\emptyset$ \\
     4 & $0$ & $(6, 2), (8, 2), (5, 1), (9, 0)$ & $\emptyset$           & $\emptyset$ \\
     5 & $2$ & $(8, 2), (5, 1), (9, 0)$         & $(13)$                  & $(0)$ \\
     6 & $2$ & $(1, 4), (3, 4), (5, 1), (9, 0)$ & $(13)$                 & $(0)$ \\
     7 & $4$ & $(3, 4), (5, 1), (9, 0)$         & $(13, 1)$              & $(0, 2)$ \\
     8 & $1$ & $(5, 1), (9, 0)$                 & $(13, 1, 8)$           & $(0, 2, 4)$ \\
     9 & $0$ & $(9, 0)$                         & $(13, 1, 8, 11)$        & $(0, 2, 4, 1)$ \\
     10 & $0$ & $(2, 2), (4, 2)$                 & $(13, 1, 8, 11)$        & $(0, 2, 4, 1)$ \\
    11 & $2$ & $(4, 2)$                         & $(13, 1, 8, 11, 3)$     & $(0, 2, 4, 1, 0)$ \\
    12 & $2$ & $\emptyset$                      & $(13, 1, 8, 11, 3, 10)$  & $(0, 2, 4, 1, 0, 2)$ \\
    \hline
\end{tabular}
\caption{Constructing the \SSA and \SLCP array.}\label{tab:step2}
\end{table}

Thus the resulting arrays are $\SSA = (13, 1, 8, 11, 3, 10)$ and $\SLCP = (0, 2, 4, 1, 0, 2)$. Indeed, the corresponding suffixes are 
\texttt{abia},
\texttt{abracadabrarabia},
\texttt{abrarabia},
\texttt{arabia},
\texttt{racadabrarabia} and
\texttt{rarabia}, which are now lexicographically sorted. Moreover, between each consecutive pair of suffixes in $\SSA$, the longest common prefixes correspond to the values in $\SLCP$ (the first value in $\SLCP$ is $0$, by definition). 

\subsection{Analysis}

We prove the following result (Theorem~\ref{thm:algo}) by analyzing the time (Lemma~\ref{lem:time}) and space (Lemma~\ref{lem:space-complexity}) complexity of the \textsc{Main-Algo}. The correctness of the algorithm  follows directly from~\cite{DBLP:conf/stacs/IKK14}.

\begin{theorem}\label{thm:algo}
For any string $T\in \Sigma^n$, any set $T_\mathcal{B}$ of $b$ 
suffixes of $T$, and any $s\in[b,n]$, \textsc{Main-Algo} with $j_\text{start}$ set to $\lfloor \log n \rfloor$ computes the
$\SSA$ and $\SLCP$ of $T_\mathcal{B}$ in $\cO(n + (bn/s) \log s)$ time using $s+7b+o(b)$ machine words. The output is correct with high probability.
\end{theorem}

\begin{lemma}\label{lem:time}
\textsc{Main-Algo} with $j_\text{start}$ set to $\lfloor \log n \rfloor$ runs in $\cO(n + (bn/s) \log s)$ time.
\end{lemma}
\begin{proof}
    The first phase of the algorithm consists of $\cO(\log n)$ iterations and a sorting step. During every iteration, each existing group is considered and every member within the group is hashed. After being hashed, it is either re-added to the same group or put into a new group. The total number of groups is at most $b - 1$, as the group structure represents a conceptual tree (hierarchy) 
    with $b$ leaves in which all internal nodes have at least two children. The number of members in each group is at most $b$. However, by amortization, it can be seen that every member (that is, every suffix and every group other than the ``root'') is processed precisely once during every iteration. Thus, we have $2b - 2 = \cO(b)$ members in total.

    For every group member, a KR fingerprint is computed. After the one-time pre-processing of $T$ in $\cO(n)$ time, the KR fingerprint of a length-$k$ substring can be computed in $\cO(\min \{k, n/s\})$ time (Lemma~\ref{lem:KRF}).
    In the first $\log s$ iterations, the cost is $\cO(n/s)$, so the total cost of these iterations is $\cO((bn/s)\log s)$.
    After $\log s$ iterations, the length $k$ of the substring whose KR fingerprint is computed 
    is $k < n/2^{\log s} = n/s$ and so the total cost of all remaining iterations is  $bn/s + bn/(2s) + bn/(4s) + \cdots +b= \cO(bn/s)$.
    Thus the total cost of computing all KR fingerprints is $\cO(n + (bn/s) \log s)$. 
    
    Every group member has its KR fingerprint taken and added to a hash table supporting constant worst-case operations with high probability~\cite{Iceberg,DBLP:conf/focs/ArbitmanNS10}. Afterwards, all members from the hash table are re-added to the groups; for every KR fingerprint collision a new group is created with its respective members, and all other members are re-added to the original group. The number of newly created groups is at most half the number of members in the original group, as every new group has to contain at least two members. So other than the fingerprinting, all operations for a single member are performed  
    in constant time with high probability,\footnote{If this is not the case, we output incorrect arrays deliberately to ensure that our algorithm is Monte Carlo.} meaning that the total time for every iteration is $\cO(b\log n)=\cO((bn/s) \log s)$.

    In the sorting step, we have two cases: (a) $b < n/ \log n$ and (b) $b \geq n/\log n$.
    The members in each group are sorted using in-place merge sort~\cite{DBLP:journals/jal/SaloweS87,DBLP:journals/njc/KatajainenPT96} (Case (a)) or in-place radix sort~\cite{DBLP:conf/esa/FranceschiniMP07} (Case (b)) in $\cO(n)$ time.
    \begin{description}
        \item[Case (a): $b < n/ \log n$.] Sorting $k$ members with merge sort takes $\cO(k \log k)$ time. Recall that there are at most $b - 1$ groups and that the total number of members over all groups is at most $2b - 2 = \cO(b)$. If the number of members to be sorted in group $i$ is $k_i$, then $k_1 + \dots + k_{b-1} = \cO(b)$ so the time needed to sort all groups is $\cO(k_1 \log k_1 + \dots + k_{b-1} \log k_{b-1}) = \cO((k_1 + \dots + k_{b-1}) \log b) = \cO(b \log b)=\cO(n)$.
        \item[Case (b): $b \geq n/\log n$.] We employ the algorithm by Franceschini et al.~\cite{DBLP:conf/esa/FranceschiniMP07}, which given an array $A$ of $k$ $\cO(\log k)$-bit integers, sorts $A$ in place in $\cO(k)$ time. Sorting the $2b-2$ members takes $\cO(b)=\cO(n)$ time, because every member can be encoded by its group id, which is a $\cO(\log b)$-bit integer, and a   
        letter, which is also a $\log \sigma = \cO(\log n)=\cO(\log b)$-bit integer, where $\sigma=|\Sigma|$.
    \end{description}

    The second phase of the algorithm is a simple stack-based DFS. Each of the $\cO(b)$ members is pushed to and popped from the stack precisely once. The further operations applied to 
    each member all take constant time, so this step takes $\cO(b)$ time.

    Adding all this together gives 
    $\cO(n) + \cO(bn/s) \cdot \cO(\log s) + \cO(n) + \cO(b) = \cO(n + (bn/s)\log s)$ time.
\end{proof}

Like the algorithm by I et al.~\cite{DBLP:conf/stacs/IKK14},
\textsc{Main-Algo} can be amended to work in $\cO(n)$ time,
when $s=b\log b$. 

\begin{lemma}
    \label{lem:space-complexity}
    \textsc{Main-Algo} can be implemented using $s + 7b + o(b)$ machine words, excluding the read-only string $T$, the array $A$ representing the set of $b$ suffixes, and the write-only output arrays $\SSA$ and $\SLCP$.
\end{lemma}
\begin{proof}

We analyze the peak space used by the algorithm neglecting the use of $\cO(1)$ machine words:
\begin{itemize}
    \item \textbf{KR fingerprints}: Pre-processing $T$ to compute KR fingerprints takes $s$ machine words by Lemma~\ref{lem:KRF}. 
    \item \textbf{Array $A$}: Array $A$ starts with $b$ integers as input, but at most $b - 1$ more integers are appended to it during the algorithm to store witness suffixes 
    for new groups, so it stores at most $b - 1$ extra integers. 
    (Even if $A$ is read-only we can simulate the append operation by using an extra array.)
    \item \textbf{Set $B$}: We implement set $B$ using three integer arrays: $K$ of size $b-1$; $C$ of size $b-1$; and $L$ of size $2b-2$. $K[i]$ stores the LCP value for the group with id $i+b$; and $L[C[i-1]+1], \ldots, L[C[i]]$  
    are the group's member id's.\footnote{If $i=1$ then the group member id's are $L[1], \ldots, L[C[i]]$.} 
    There are at most $b - 1$ groups; for every group one integer is stored in $K$ and $C$ as well as the group member id's in $L$. The total number of group members is at most $2b - 2$, since all groups except the ``root'' group are a member. Thus set $B$ can be implemented using $4b - 4$ integers.
    
    \item \textbf{Hash table $H_i$}: While processing group $i+b$, every group member is (removed from the group and) added to a hash table $H_i$ as satellite value of the corresponding KR fingerprint key.   
    We use a space-efficient hash table storing $c_i=C[i]-C[i-1]$ integers (KR fingerprints) as keys: By using~\cite{Iceberg,DBLP:conf/focs/ArbitmanNS10}, we implement $H_i$ using $(1+\epsilon)c_i$ machine words, for any $\epsilon =\Omega(\log\log c_i/\log c_i)$. We need at most $b$ integers to maintain the size of the satellite values per KR fingerprint because every group can have at most $b$ members. By choosing $\epsilon=\log\log c_i/\log c_i$ we need at most $2b+o(b)$ machine words in total to maintain the hash table.
\end{itemize}

    We can delete 
    the fingerprint data structure and the hash table before moving to the sorting step.
    Sorting does not use any additional space because merge sort and radix sort can be implemented in-place~\cite{DBLP:journals/jal/SaloweS87,DBLP:journals/njc/KatajainenPT96,DBLP:conf/esa/FranceschiniMP07}, thus using only $\cO(1)$ additional machine words. The first phase of the algorithm uses at most $s + 7b + o(b)$ machine words but at the end of it we have $5b + \cO(1)$ machine words stored: array $A$ and set $B$.

We now analyze the space used in the second phase (Algorithm~\ref{alg:step2}); in particular, the space taken by the search stack. The stack stores at most every group and every suffix. However, the stack never simultaneously stores a member and one of its ancestors, meaning the maximum size of the stack at any point is at most the maximum width of the sparse suffix tree, which is $b$. Every element in the stack consists of two integers, so the stack takes up at most $2b$ machine words. No other machine words need to be stored as the maximum stack size $b$ is known in advance and so the stack 
is implemented using an array.

Adding this together gives at most $s + 7b + o(b)$  machine words in total needed for the algorithm.
\end{proof}

\section{A Simple Parameterized Algorithm}
\label{sec:improved}

In real-world datasets, the $b$ suffixes in $T_\mathcal{B}$ will generally not share very long prefixes. 
Even when they do, it is highly unlikely that all of them have this property. While \textsc{Main-Algo} is theoretically efficient, it would waste a lot of time with such datasets by considering large overlaps between suffixes when in reality the longest common prefixes are much shorter or when only very few suffixes share very long prefixes. Below, we show a simple method to take advantage of this, by only considering short common prefixes in the beginning and then extending them \emph{only} for the suffixes that happen to share longer prefixes. By considering an extra parameter $b'$ 
indicating the number of suffixes that share longest common prefixes longer than a certain threshold, we arrive at a time and space complexity that appears favorable for such real-world datasets.

\paragraph{Main Idea.} We design an algorithm for constructing \SSA and \SLCP which is parameterized by the total number $b'\leq b$ of suffixes which have an LCP value of at least $\ell= 2^{\lfloor \log \frac{n}{b} \rfloor + 1} - 1$ with some other suffix. We show that partitioning the $b$ suffixes into two classes (one with suffixes with LCP value strictly less than $\ell$; and another with suffixes with LCP value greater than or equal to $\ell$) can be done in $\cO(n)$ time. In particular, we show that it suffices to invoke Theorem~\ref{thm:algo} twice: once (with a small change) for the $b$ suffixes; and once (as is) for the $b'$ suffixes; and then merge the partial results in $\cO(b)$ time to obtain the final $\SSA$ and $\SLCP$ array.

\paragraph{Description and Pseudocode.} The pseudocode for the algorithm is given as \textsc{Parameterized-Algo} (Algorithm~\ref{alg:improved}).\footnote{The pseudocode is complete in the sense that it only assumes the implementation of \textsc{Main-Algo}.} A line-by-line explanation of the algorithm follows.

\textsc{Parameterized-Algo} invokes the original algorithm \textsc{Main-Algo} twice with different arguments. In Line 1, it calls \textsc{Main-Algo} with the full array $A$ as argument (and $s = b$). We set the parameter $j_\text{start}$ that indicates the starting value of $j$ (Line 5 of Algorithm~\ref{alg:step1}) to $\lfloor \log \frac{n}{b} \rfloor$, meaning that $j$ starts at a lower value than the value $\lfloor \log n \rfloor$ used in the \textsc{Main-Algo} 
and so it will take less time to complete. The result of this is that $\SSA$ will only be sorted up to $\ell = 2^{\lfloor \log \frac{n}{b} \rfloor + 1} - 1$ positions. This means that for every consecutive pair of suffixes in $\SSA$, if their LCP value is less than $\ell$, they will already be sorted correctly, whereas the other suffixes, with associated LCP values of $\ell$, will need to be further sorted in the second phase (Lines 8 to 13) of \textsc{Parameterized-Algo}.

What remains is to identify the suffixes that need to be further sorted,
sort these suffixes separately from the others, and re-insert them into the output arrays along with the corrected LCP values. We use two arrays $A'$ and $P$ for this purpose: $A'$ contains the suffixes; and $P$ tracks the positions in $\SSA$ that these suffixes are taken from, to ensure that they will later be re-inserted at the correct positions. In Line 5, we ensure that the right suffixes are tracked in these arrays, namely those that have an LCP value of $\ell$ with their predecessor or successor suffix.  
 If any such suffixes are found, we invoke \textsc{Main-Algo} again (Line 9), 
 but with just these suffixes (those in array $A'$) as input, and with the default value of $j_\text{start}=\lfloor \log n \rfloor$.  
 This means that the suffixes of $A'$ will now be 
 fully sorted rather than being sorted up to $\ell$ positions. Then, in Lines 10 and 11, we insert these re-sorted suffixes at the same positions that they were taken from before, but in the corrected order. In Lines 12 and 13, we also copy the associated LCP values, but only at the positions in-between two re-sorted suffixes, as all other LCP values were already correct. 
 
 We next state and prove Theorem~\ref{the:param}.
 
\begin{theorem}\label{the:param}
For any string $T\in \Sigma^n$ and any set $T_\mathcal{B}$ of $b$ 
suffixes of $T$, \textsc{Parameterized-Algo} computes
the $\SSA$ and $\SLCP$ of $T_\mathcal{B}$ in $\cO(n + (b'n/b) \log b)$ time using $8b + 4b' + o(b)$ machine words, where $b'$ is the total number of $i$ such that $\SSA[i] \in \mathcal{B}$ and $\SLCP[i]\geq \ell$ or $\SLCP[i+1] \geq \ell$, with $\ell=2^{\lfloor \log \frac{n}{b} \rfloor + 1} - 1$. The output is correct with high probability. When $b'=\cO(b/\log b)$, \textsc{Parameterized-Algo} runs in $\cO(n)$ time using $8b + o(b)$ machine words. 
\end{theorem}

\begin{algorithm}[H]\footnotesize
    \caption{\textsc{Parameterized-Algo}}
     \label{alg:improved}
    \textbf{Input:} string $T \in \Sigma^n$,  integer $b$, and array $A$ of $b$ distinct integers from $[n]$\\ 
     \textbf{Output:} $\SSA$ and $\SLCP$
    \begin{algorithmic}[1]
        \State{$\SSA, \SLCP \gets \textsc{Main-Algo}(T, A, b, j_\text{start} = \lfloor \log \frac{n}{b} \rfloor)$} \Comment{Quasi-sort the $b$ suffixes} 
        \State{$\ell \gets 2^{\lfloor \log \frac{n}{b} \rfloor + 1} - 1$}
        \State{$P, A' \gets \text{empty arrays}$}
        \For{$i = 1,\ldots,b$} \Comment{Check which of the $b$ suffixes need further sorting}
                \If{$\SLCP[i] = \ell \lor ( i < b \land \SLCP[i + 1] = \ell)$} \Comment{Suffix $T[\SSA[i]\dd n]$ needs further sorting}
                    \State{$P.\text{append}(i)$} 
                    \State{$A'.\text{append}(\SSA[i])$} 
                \EndIf
        \EndFor
        \If{$|A'| > 0$}
            \State{$\SSA', \SLCP' \gets \textsc{Main-Algo}(T, A', |A'|, j_\text{start}=\lfloor \log n \rfloor)$} \Comment{Sort $|A'|$ out of $b$ suffixes}
            \For{$i = 1,\ldots,|A'|$} \Comment{Merge the partial arrays}
                \State{$\SSA[P[i]] \gets \SSA'[i]$}
                \If{$\SLCP[P[i]] = \ell$} \Comment{Only copy the $\SLCP$ value if the suffix at $P[i] - 1$ was also re-sorted}
                    \State{$\SLCP[P[i]] \gets \SLCP'[i]$}
                \EndIf
            \EndFor
        \EndIf
        \State{\Return{\SSA and \SLCP}}
    \end{algorithmic}
\end{algorithm}

\paragraph{Time Complexity.} The first phase of the algorithm (Line 1) runs in $\cO(\log \frac{n}{b})$ iterations. The longest prefixes whose KR fingerprints are computed have length $\cO(\frac{n}{b})$, and there are $\cO(b)$ KR fingerprints computed in each iteration. This means that computing the KR fingerprints during the first phase takes $\cO(b) \cdot (\cO(\frac{n}{b}) + \cO(\frac{n}{2b}) + \cO(\frac{n}{4b}) + \ldots) = \cO(n)$ time. Hashing the fingerprints takes $\cO(b \log \frac{n}{b}) = \cO(n)$ worst-case time in total with high probability. (Grouping the fingerprints via distribute-and-collect, like the algorithm by I et al.~\cite{DBLP:conf/stacs/IKK14}, would incur a multiplicative factor of $\log_s n$.) Sorting takes $\cO(n)$ time (see Lemma~\ref{lem:time}).
Therefore the entire first phase runs in $\cO(n)$ time. The second phase (Lines 8 to 13) computes KR fingerprints of longer prefixes as well and otherwise runs the same as \textsc{Main-Algo}, with the exception that only $b'$ suffixes are now sorted. By Lemma~\ref{lem:time}, for $s=b$, this takes $\cO(n + (b'n/b)\log b)$ time. All other operations run in single loops over arrays of size $b$ or $b'$ with constant-time operations, and thus take $\cO(b)$ time. Adding everything together gives $\cO(n + (b'n/b) \log b)$ time. When $b' = \cO(b/\log b)$, the running time becomes $\cO(n)$.

\paragraph{Space Complexity.} The first phase of the algorithm uses $s + 7b + o(b)$ machine words as per Lemma~\ref{lem:space-complexity}. 
The additional arrays $P$, $A'$, $\SSA'$ and $\SLCP'$ use $4b'$ machine words in total. The second invocation of \textsc{Main-Algo} uses $s+ 7b' + o(b')$ machine words as per Lemma~\ref{lem:space-complexity}. 
By setting $s=b$, the algorithm uses $8b + 4b' + o(b)$ machine words in total. 
If $b'=\cO(b/\log b)$, the algorithm uses $8b + o(b)$ machine words.

\paragraph{Correctness.} The correctness of $\SSA$ is proved by Lemma~\ref{lem:improved-3} and that of $\SLCP$ is proved by Lemma~\ref{lem:improved-4}. To prove these lemmas, we first show the auxiliary Lemmas~\ref{lem:improved-1} and~\ref{lem:improved-2}. 

\begin{lemma}
\label{lem:improved-1}
    Let $\SSA_1$ be the instance of $\SSA$ after the first invocation of \textsc{Main-Algo} (Line 1). The strings $T[\SSA_1[i] \dd n]$, $i\in[1,b]$, are sorted up to their prefix of length $\ell = 2^{\lfloor \log \frac{n}{b} \rfloor + 1} - 1$.
\end{lemma}
\begin{proof}
    In \textsc{Main-Algo}, all LCP values can be increased by powers of two in each iteration. With the starting value $j_\text{start} = \lfloor\log \frac{n}{b}\rfloor$,
    this adds up to a maximum LCP value of $\ell$ in any group. At any point during \textsc{Main-Algo}, two suffixes that are in the same group with LCP value $k$ share a longest common prefix of length at least $k$. Thus, this invocation of \textsc{Main-Algo} will compute the LCP values between suffixes correctly if they are at most $\ell$, and all other LCP values will be $\ell$. The sorting step takes into account only the letter which appears after the computed (longest) common prefix, so if the LCP between any two suffixes is less than $\ell$ the suffixes will be sorted correctly. The statement follows.
\end{proof}

\begin{lemma}
\label{lem:improved-2}
Let $\SSA_1$ be the instance of $\SSA$ after the first invocation of \textsc{Main-Algo} (Line 1), and let $\SSA_2$ be the instance of $\SSA$ returned at the end of \textsc{Parameterized-Algo} (Line 14). For every $i \in [1,b]$, either $\SSA_1[i]= \SSA_2[i]$ and $\SSA_1[i]$ and $\SSA_2[i]$ have a longest common prefix of length $n-\SSA_1[i]+1$,  
or $\SSA_1[i]\neq \SSA_2[i]$ and $\SSA_1[i]$ and $\SSA_2[i]$ have a longest common prefix of length at least $\ell$.
\end{lemma}
\begin{proof}
    If $\SSA_1[i] = \SSA_2[i]$, this is trivial, so we only concern ourselves with the case $\SSA_1[i] \ne \SSA_2[i]$. 
    In this case, the value was overwritten in  Line 11, meaning that the suffix $\SSA_1[i]$ was stored in $A'$ in Line 7 to be re-sorted in the second invocation of \textsc{Main-Algo}. The same must hold for $\SSA_2[i]$.

    Consider the array $A'$ as it is built in Lines 4-7. By Lemma~\ref{lem:improved-1}, the suffixes of $\SSA_1$ are sorted up to their length-$\ell$ prefix; since the entries of $A'$ appear in the same order as they appear in $\SSA_1$, this must also be the case for $A'$. Because the suffixes of $A'$ are already sorted correctly up to their length-$\ell$ prefix, it must be that for every position $j \in [1,b']$, $A'[j]$ and $\SSA'[j]$ have the same length-$\ell$ prefix. Now note that if $\SSA_1[i]$ appears in position $j$ in $A'$, then $\SSA_2[i]$ will take the value from $\SSA'[j]$. Since $A'[j]$ and $\SSA'[j]$ have a length-$\ell$ common prefix, $\SSA_1[i]$ and $\SSA_2[i]$ must as well.
\end{proof}

\begin{lemma}
\label{lem:improved-3} 
    The instance $\SSA_2$ of $\SSA$ returned at the end of \textsc{Parameterized-Algo} (Line 14), 
    contains the suffixes of $A$ sorted lexicographically.
\end{lemma}
\begin{proof}
    We prove this by showing that for any two consecutive positions $i$ and $i + 1$, $\SSA_2[i]$ and $\SSA_2[i + 1]$ appear in the right order. 
    Let $\SSA_1$ be the instance of $\SSA$ after the first invocation of \textsc{Main-Algo}.

    We already know that $\SSA_1$ is sorted correctly up to $\ell$ positions. This means that for any $i$, if the longest common prefix of $\SSA_1[i]$ and $\SSA_1[i + 1]$ is shorter than $\ell$, they already appear in the correct order in this array. If neither suffix is overwritten after the second phase, this is also trivially the case for them in $\SSA_2$.

    Now suppose that exactly one of the two (wlog $\SSA_1[i + 1]$) is replaced by some other suffix $s$ while the other remains the same. Let $k$ be the LCP of $\SSA_1[i]$ and $\SSA_1[i + 1]$. By Lemma~\ref{lem:improved-2}, $\SSA_1[i + 1]$ and $s$ have a longest common prefix of length at least $\ell$. This is longer than $k$, which is strictly less than $\ell$. This means that the $(k + 1)$-th letter of $s$ is the same as that of $\SSA_1[i + 1]$, which is the first position in which it differs from $\SSA_1[i]$. Thus $\SSA_2[i] = \SSA_1[i]$ and $\SSA_2[i + 1] = s$ are sorted correctly relative to one another.

    The remaining case is when $\SSA_1[i]$ and $\SSA_1[i + 1]$ have a longest common prefix of length $\ell$ or longer. In this case, both suffixes are added to $A'$ to be re-sorted in the second invocation, and both $\SSA_2[i]$ and $\SSA_2[i + 1]$ may take the value of another suffix. The second invocation of the main algorithm sorts all suffixes in $A'$ completely, returning $\SSA'$. The suffixes in $\SSA'$ are then re-inserted into $\SSA_2$, in which they will appear in the same order as they did in $\SSA'$. Therefore, no matter which suffixes end up at $\SSA_2[i]$ and $\SSA_2[i + 1]$, they also appeared consecutively in $\SSA'$ and therefore must be sorted correctly.    
\end{proof}

\begin{lemma}
\label{lem:improved-4}
    For any two consecutive positions $i$ and $i + 1$, $\SLCP[i + 1]$, as returned by Algorithm~\ref{alg:improved}, gives the length of the longest common prefix of $\SSA[i]$ and $\SSA[i + 1]$.
\end{lemma}
\begin{proof}
    Let $\SSA_1$ and $\SLCP_1$ be the arrays returned by the first invocation of \textsc{Main-Algo}, and $\SSA_2$ and $\SLCP_2$ the arrays produced at the end. By Lemma~\ref{lem:improved-1}, if $\SLCP_1[i + 1] < \ell$, this value is correct. Therefore, the only values that need to be overwritten for $\SLCP_2$ are when $\SLCP_1[i + 1] = \ell$. The check at Line 12 ensures this. Of course, when $\SLCP_1[i + 1] = \ell$, then both $\SSA_1[i]$ and $\SSA_1[i + 1]$ are added to $A'$ in order to be re-sorted in the second invocation. The values at $\SSA_2[i]$ and $\SSA_2[i + 1]$ are then replaced by two suffixes that appear consecutively in $\SSA'$, say $\SSA'[j]$ and $\SSA'[j + 1]$. By the correctness of \textsc{Main-Algo}, the LCP value of these two suffixes is given by $\SLCP'[j + 1]$, which is the value that $\SLCP_2[j + 1]$ takes.
\end{proof}

\paragraph{Random Strings.} Finally, we show that \textsc{Parameterized-Algo}
    can be trivially amended to work in $\cO(n)$ time for any string chosen uniformly at random from $\Sigma^n$.

\begin{theorem}\label{the:random}
For any string $T$ chosen uniformly at random from $\Sigma^n$ and any set $T_\mathcal{B}$ of $b$ suffixes of $T$, $\SSA$ and $\SLCP$ of $T_\mathcal{B}$ can be computed in $\cO(n)$ time using $\cO(b)$ space. The output is correct with high probability. 
\end{theorem}

\begin{proof}
    We assume $|\Sigma|\geq 2$, otherwise the problem has a trivial solution.
    Bollob{\'{a}}s and Letzter~\cite[Theorem 4]{DBLP:journals/ejc/BollobasL18} showed that the maximum length of an LCE on $T$ is at most $2\log_{|\Sigma|}n+\log_{|\Sigma|}\log_{|\Sigma|}n$ with high probability. We bound this from above by $3\log n$ and amend \textsc{Parameterized-Algo} as follows:

\begin{description}
    \item[Case (a): $b \log n < n$.] We invoke \textsc{Main-Algo} by setting $j_\text{start}$ to the smallest integer such that $2^{ j_\text{start}} \ge 2\lfloor \log n \rfloor$, which gives $\ell=2\lfloor \log n\rfloor\cdot 2-1=4\lfloor\log n\rfloor-1$. 
    After the $\cO(n)$-time preprocessing of Lemma~\ref{lem:KRF}, computing the KR fingerprints takes $\cO(b) \cdot 4(\cO(\frac{\log n}{1}) + \cO(\frac{\log n}{2}) + \cO(\frac{\log n}{4}) + \ldots) = \cO(b \log n)$ time. 
    Hashing the fingerprints takes $\cO(b)$ time per iteration with high probability, and so $\cO(b \log n)$ total time. Merge sort takes $\cO(b\log b)$ time.
    Since $\ell>3\log n$, all suffixes of $T_\mathcal{B}$ will be fully sorted from the first invocation of \textsc{Main-Algo}. If $b'=\cO(b/\log b)$ suffixes are still unsorted after the first invocation, these will be fully sorted in the second invocation of \textsc{Main-Algo} in $\cO(n)$ time (Theorem~\ref{the:param}). If $b'=\omega(b/\log b)$, we output incorrect arrays. The total time complexity is $\cO(n+b\log n)=\cO(n)$.
   \item[Case (b): $b \log n \geq n$.] Assume that we have $\cO(s)$ space to sort the $b$ suffixes; we can do it efficiently using radix sort because it suffices to sort all prefixes of them of length $\cO(\log_{\sigma} n)$ by the Bollobas and Letzter's result, where $\sigma=|\Sigma|$ (otherwise, we output incorrect arrays).
    The $b$ prefixes are each of length at most $c \log_{\sigma} n$, for some $c=\cO(1)$; so radix sort takes 
    $\cO((b + s) (c \cdot \log n/\log \sigma) \cdot (\log \sigma/\log s))$ time, because we have at most $(c \log n/\log \sigma)$ letters in every prefix, and each time we sort $b$ letters, one from each prefix, we use $(\log \sigma/\log s)$ rounds of counting sort. 
    Conveniently, the $\log \sigma$ terms cancel out. Then, because we set $s = b$, and by the fact that we are in the case $b \geq n/ \log n$, we have that $\log n/\log s = \cO(1)$. The total time complexity is thus $\cO(b + s) = \cO(b)$. The total space used is $\cO(s)=\cO(b)$.
    By comparing adjacent suffixes we compute the $\SLCP$ array within the same complexities.
\end{description}
    
\end{proof}

\section{Experiments}\label{sec:experiments}

\paragraph{Implementations.} We have used the following \texttt{C++} implementations that we make all available at \url{https://github.com/lorrainea/SSA} under GPL-3.0 license. Although these implementations are somewhat simplified for practical efficiency purposes, they are quite close to their original theoretical descriptions. 
\begin{enumerate}
\item 
A simplified implementation of the algorithm of Prezza~\cite{DBLP:journals/talg/Prezza21}, written by the author.\footnote{\url{https://github.com/nicolaprezza/rk-lce/blob/master/sa-rk.cpp}} 
It first constructs an LCE data structure in $\cO(n)$ time in place: using exactly the same space as $T$ (plus $\cO(1)$ machine words). Since the structure supports $\cO(\log n)$-time LCE queries, the implementation uses LCE queries and \texttt{std::sort} to construct the \SSA. We have amended this implementation to also compute the \SLCP in $\cO(b\log n)$ extra time using the \SSA and LCE queries. We denote this implementation by \textsc{SSA-LCE}.
Note that, for example, the use of \texttt{std::sort} does not guarantee that sorting the suffixes is done in-place in practice.
\item A simplified implementation of \textsc{Main-Algo}, written by us.
It is simplified in the sense that for hashing we use
the \texttt{unordered\_dense} class.\footnote{\url{https://github.com/martinus/unordered_dense}} Also if the number of KR fingerprints to be grouped is smaller than a predefined constant (set to 1.5M by default), then we resort to \texttt{std::sort} to achieve the same goal.
Furthermore, we have implemented array $A$ and set $B$ by means of \texttt{std::vector}, and thus we neither had full control of the exact number of machine words used per vector, nor did we know when these became available upon memory deallocation.
For KR fingerprints, we have used the class by Kempa.\footnote{\url{https://github.com/dominikkempa/lz77-to-slp/blob/main/src/karp_rabin_hashing.cpp}}
We denote this implementation by \textsc{MA}.
\item An implementation produced by us of \textsc{Parameterized-Algo} using the above simplified implementation of \textsc{Main-Algo}. We denote this implementation by \textsc{PA}.
\end{enumerate}

\begin{figure}[!ht]
     \centering
        \includegraphics[width=0.5\textwidth]{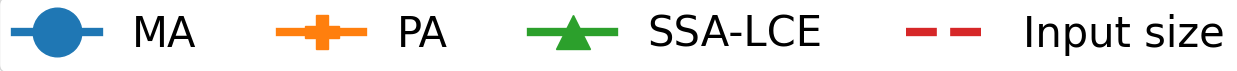}
        
        \subfloat[][\centering Time for varying $n$\label{time:small:fastq:n}] 
        {\includegraphics[width=0.23\textwidth]{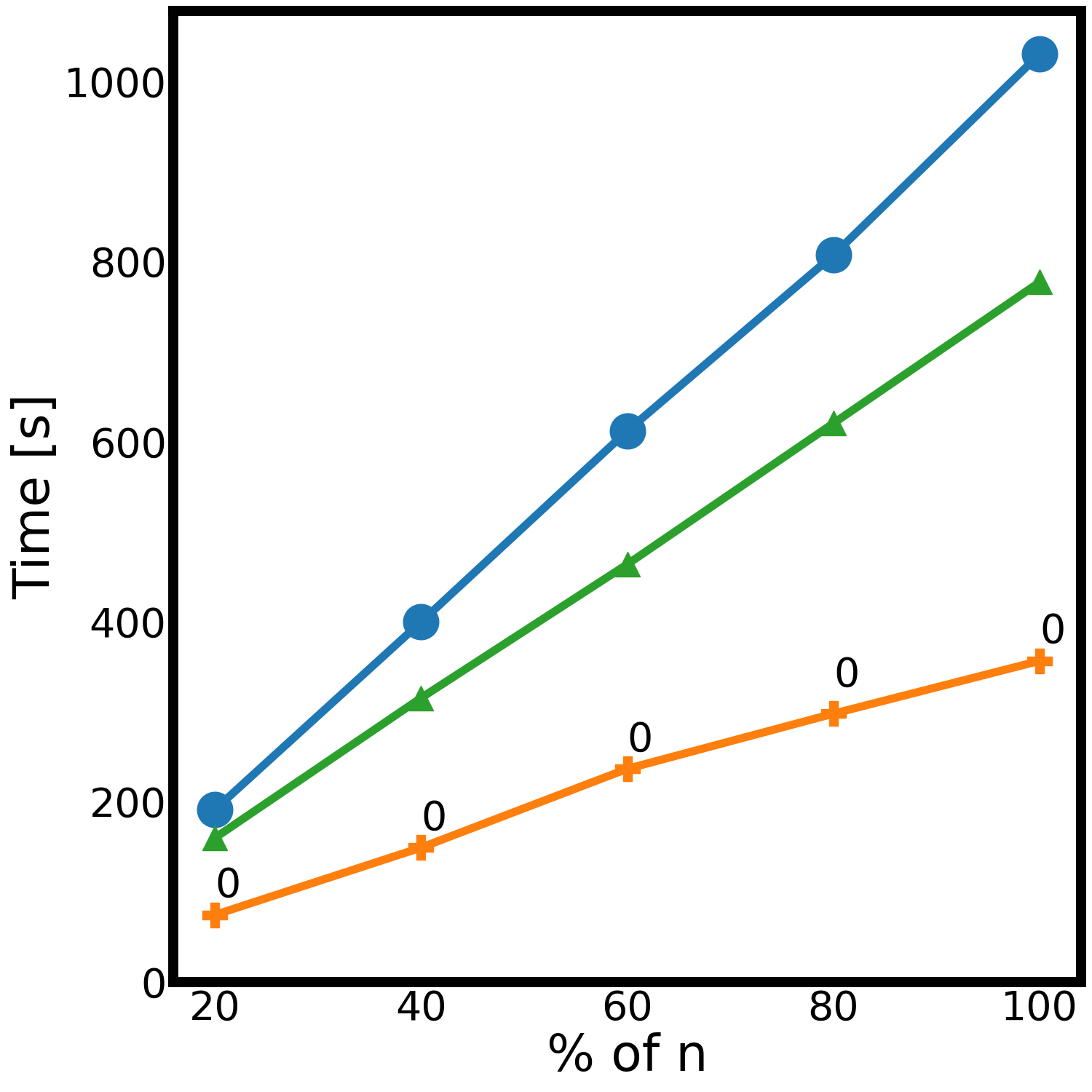}}\hspace{1mm}
        \subfloat[][\centering Time for varying $b$\label{time:small:fastq:b}]{\includegraphics[width=0.23\textwidth]{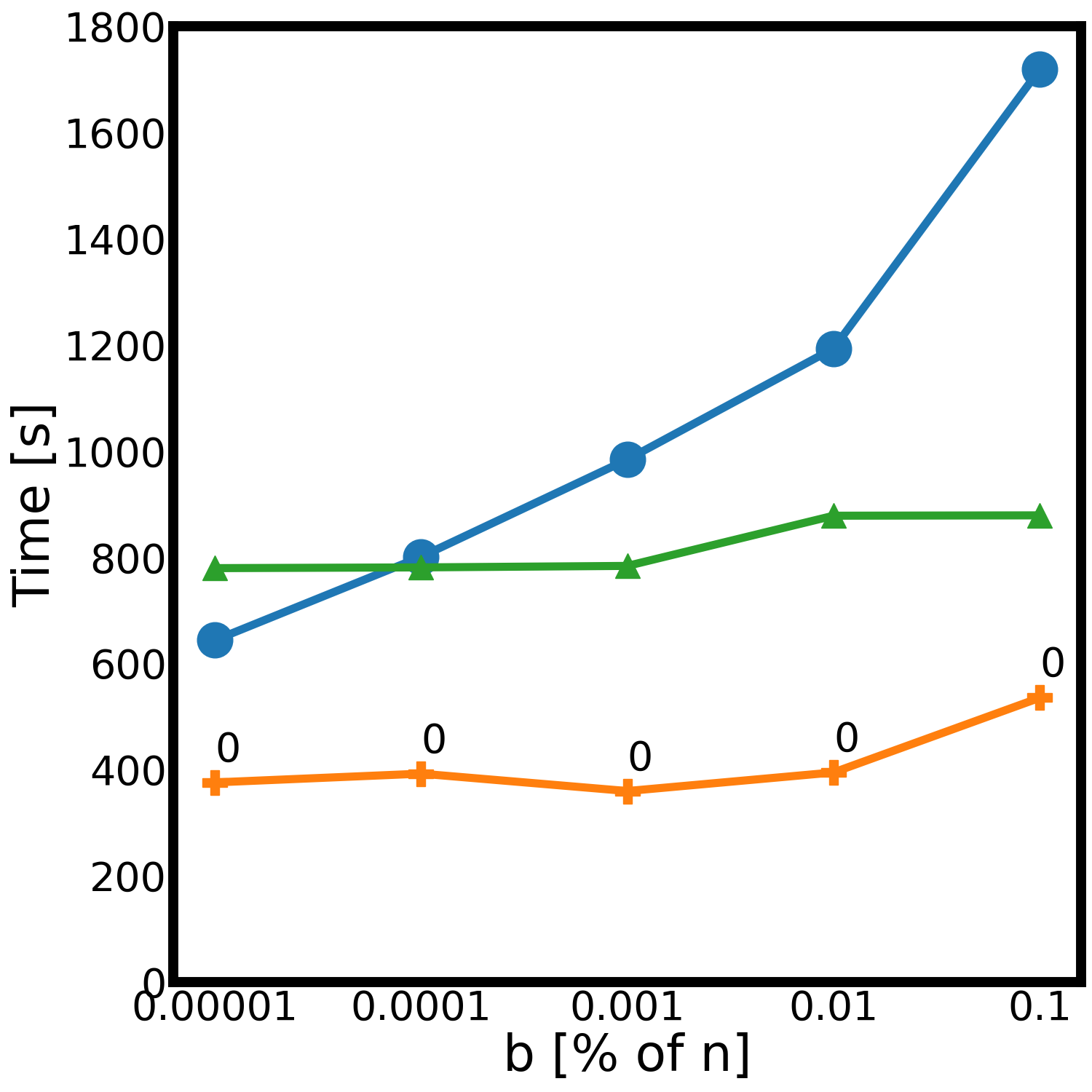}}\hspace{1mm}
        \subfloat[][\centering Memory for varying $n$ \label{mem:small:fastq:n}]{\includegraphics[width=0.23\textwidth]{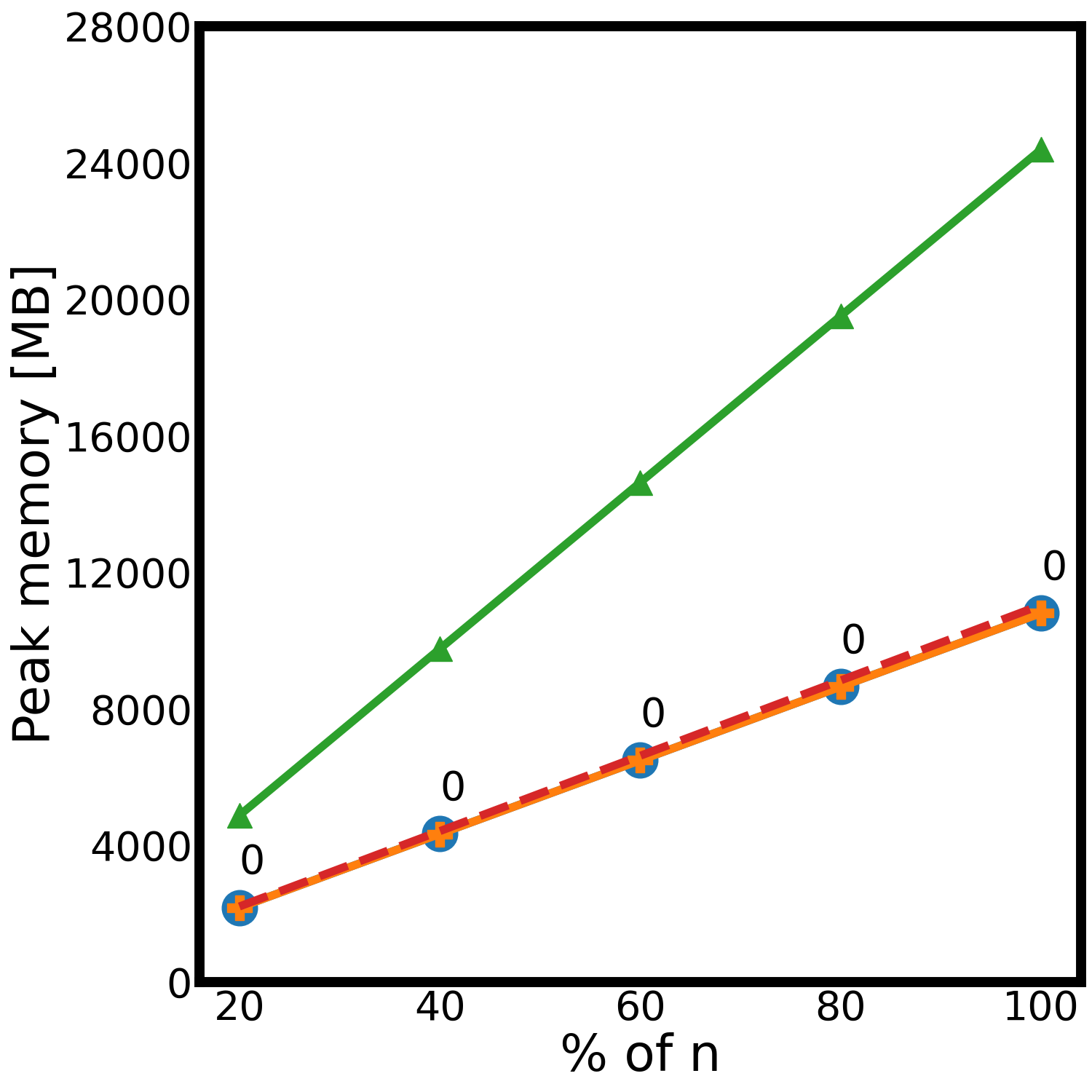}}\hspace{1mm}
        \subfloat[][\centering Memory for varying $b$ \label{mem:small:fastq:b}]{\includegraphics[width=0.23\textwidth]{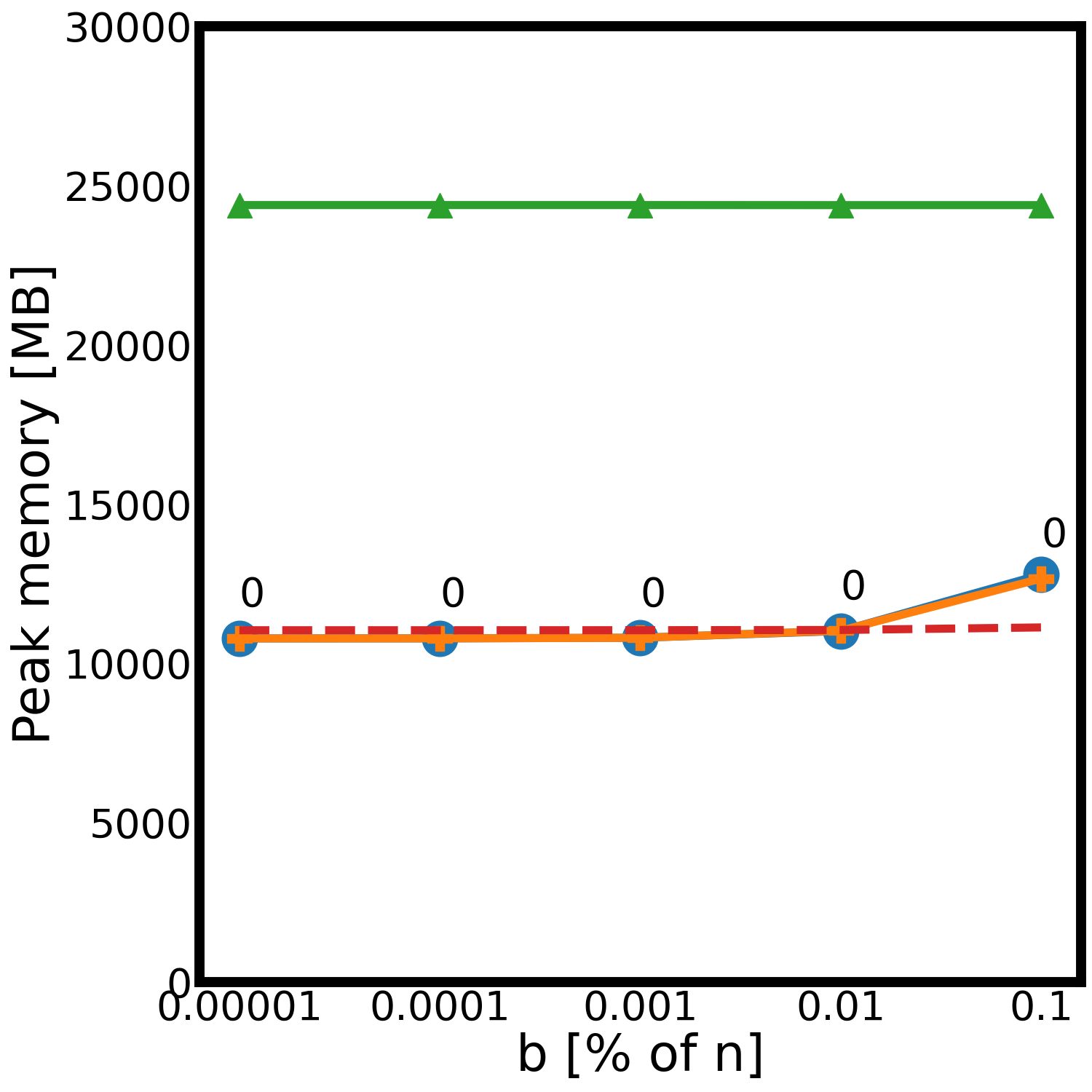}}

     \centering 
        \subfloat[][\centering Time for varying $n$\label{time:small:amazon:n}]{\includegraphics[width=0.23\textwidth]{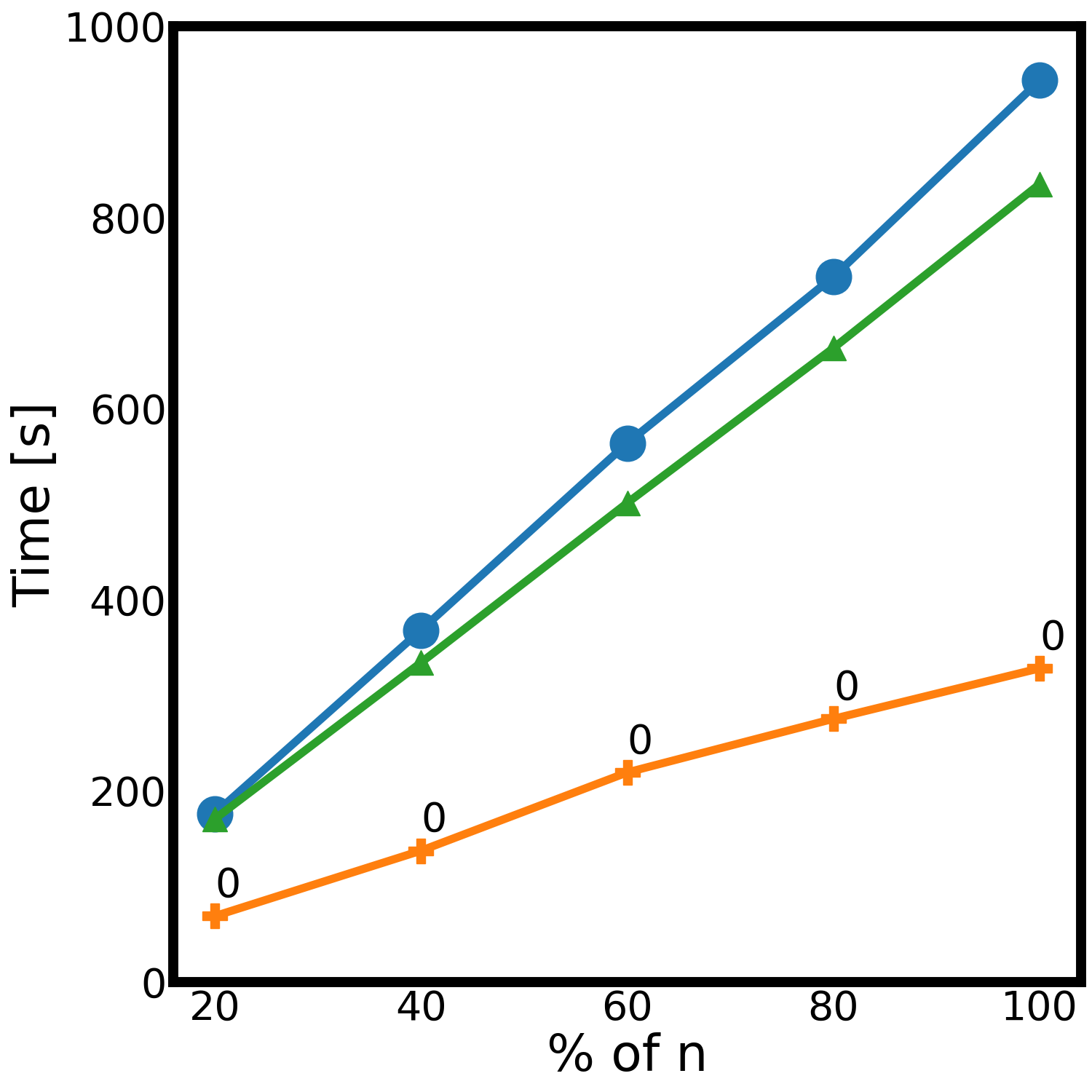}}\hspace{1mm}
        \subfloat[][\centering Time for varying $b$\label{time:small:amazon:b}]{\includegraphics[width=0.23\textwidth]{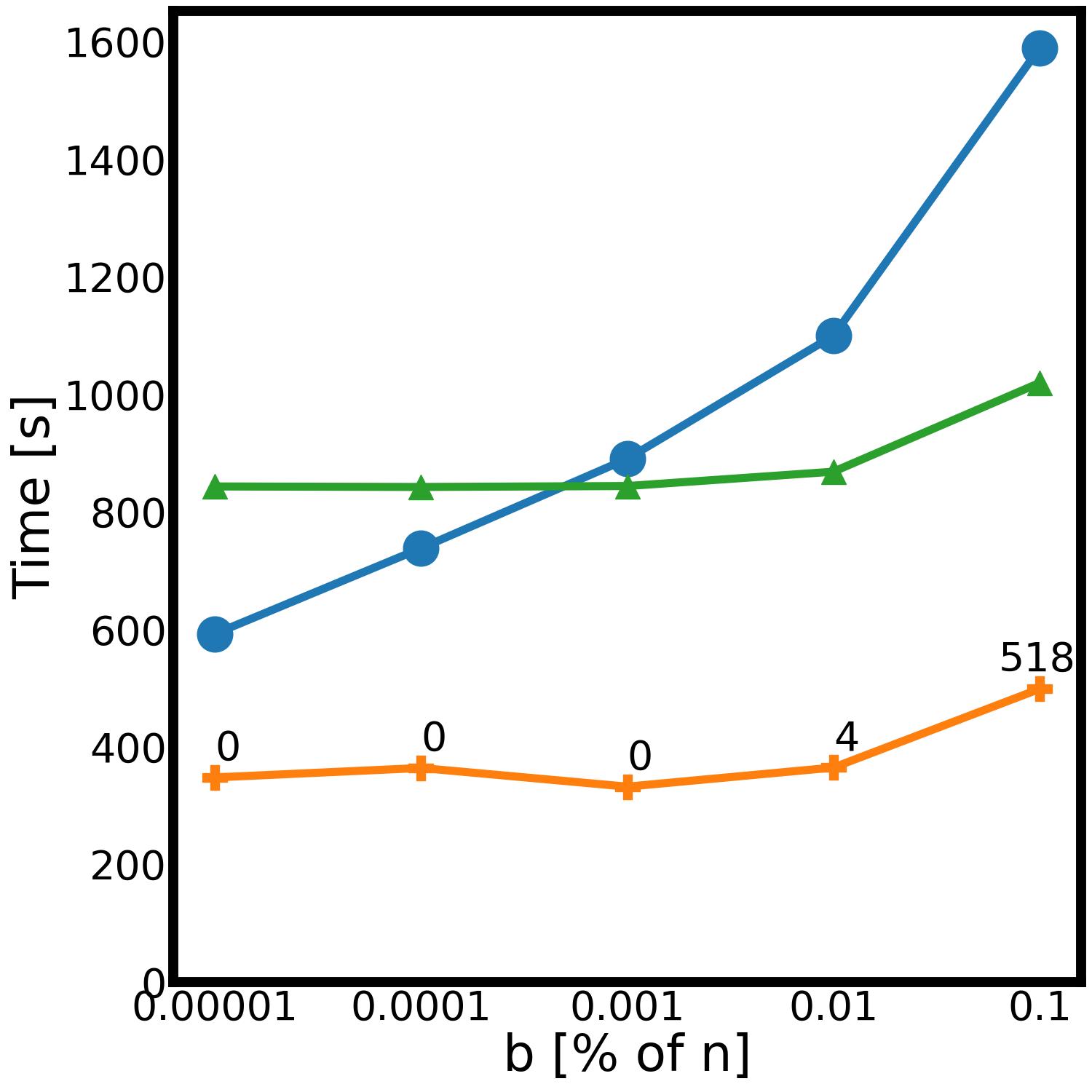}}\hspace{1mm}
        \subfloat[][\centering Memory for varying $n$ \label{mem:small:amazon:n}]{\includegraphics[width=0.23\textwidth]{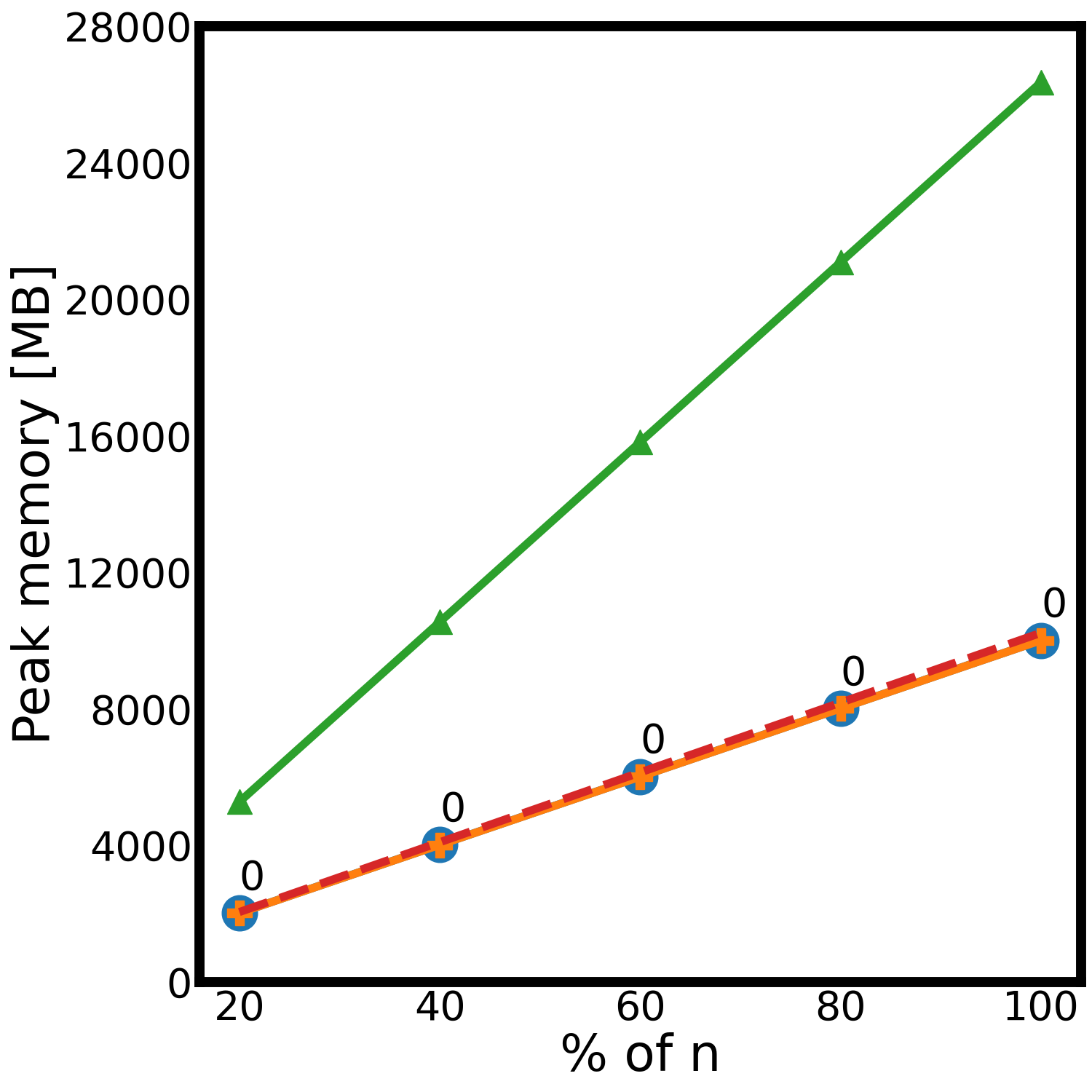}}\hspace{1mm}
        \subfloat[][\centering Memory for varying $b$ \label{mem:small:amazon:b}]{\includegraphics[width=0.23\textwidth]{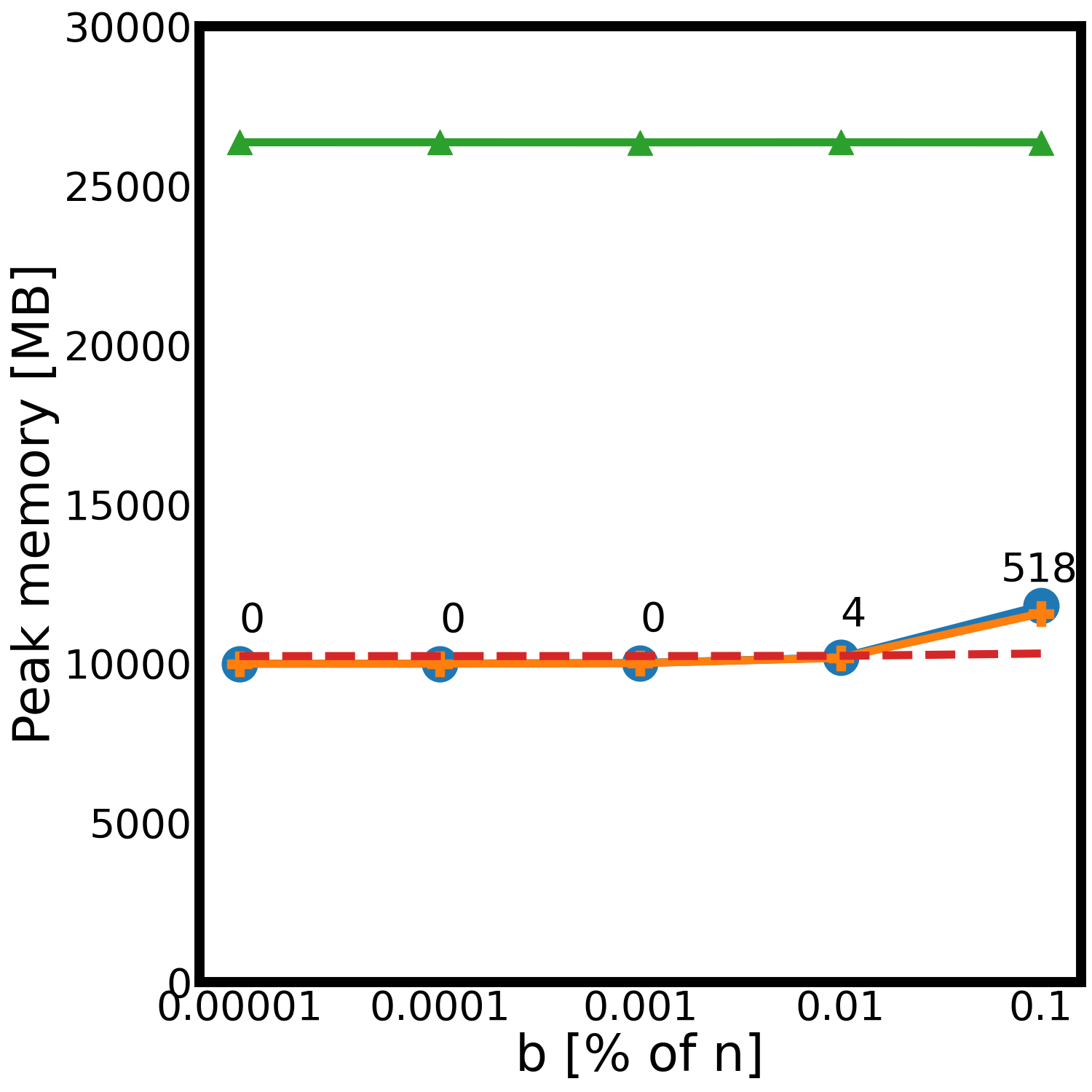}}

      \centering 
        \subfloat[][\centering Time for varying $n$\label{time:small:random:n}]{\includegraphics[width=0.23\textwidth]{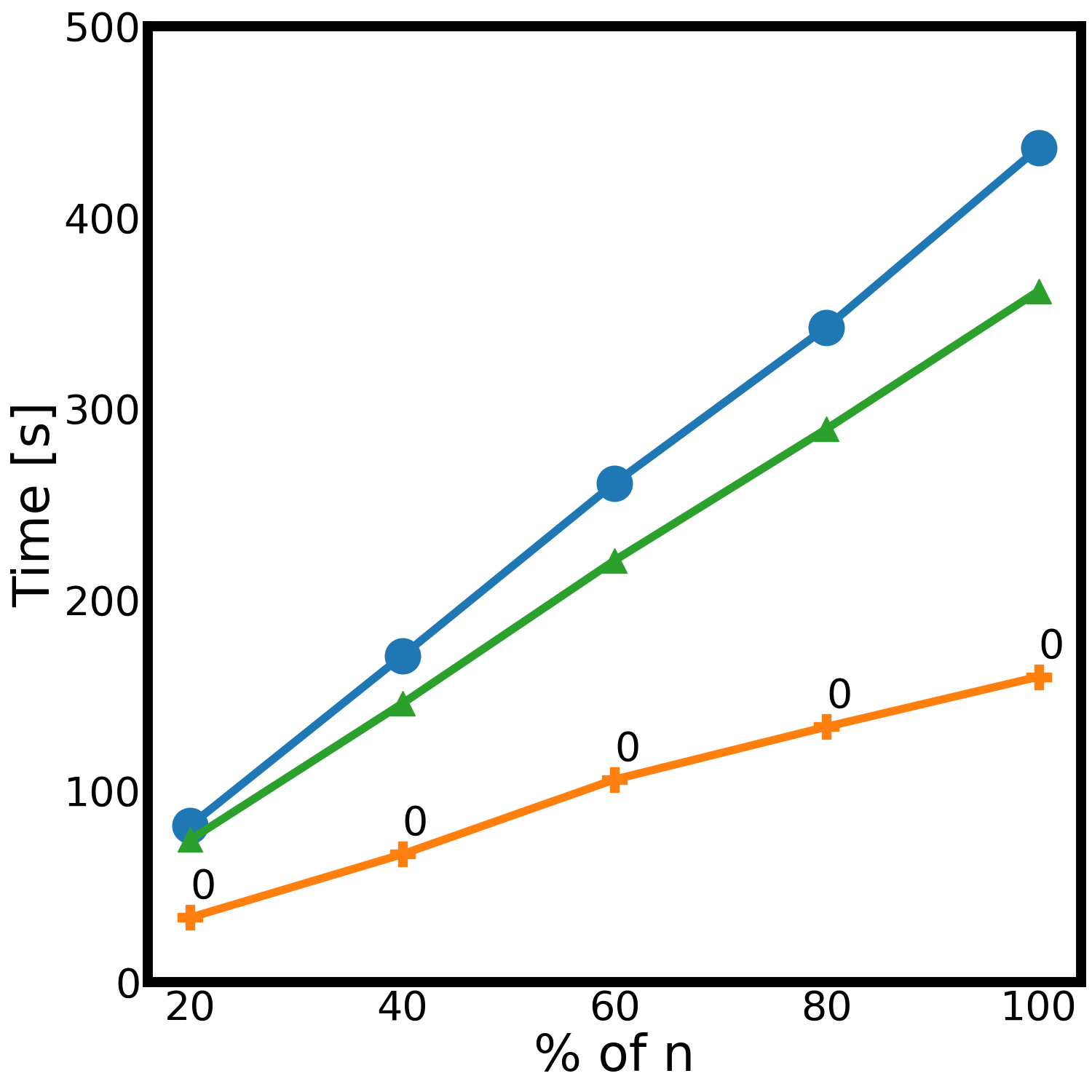}}\hspace{1mm}
        \subfloat[][\centering Time for varying $b$\label{time:small:random:b}]{\includegraphics[width=0.23\textwidth]{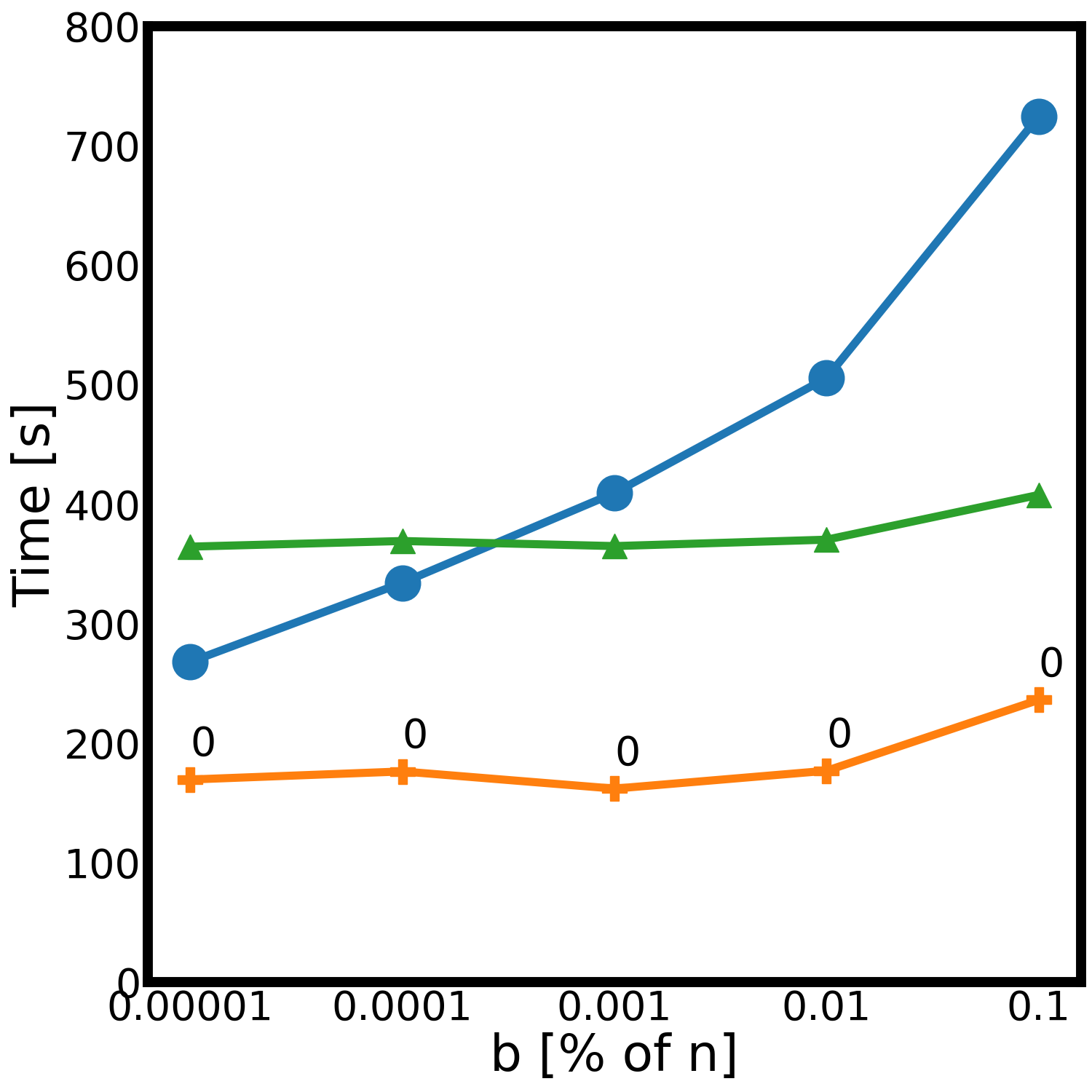}}\hspace{1mm}
        \subfloat[][\centering Memory for varying $n$ \label{mem:small:random:n}]{\includegraphics[width=0.23\textwidth]{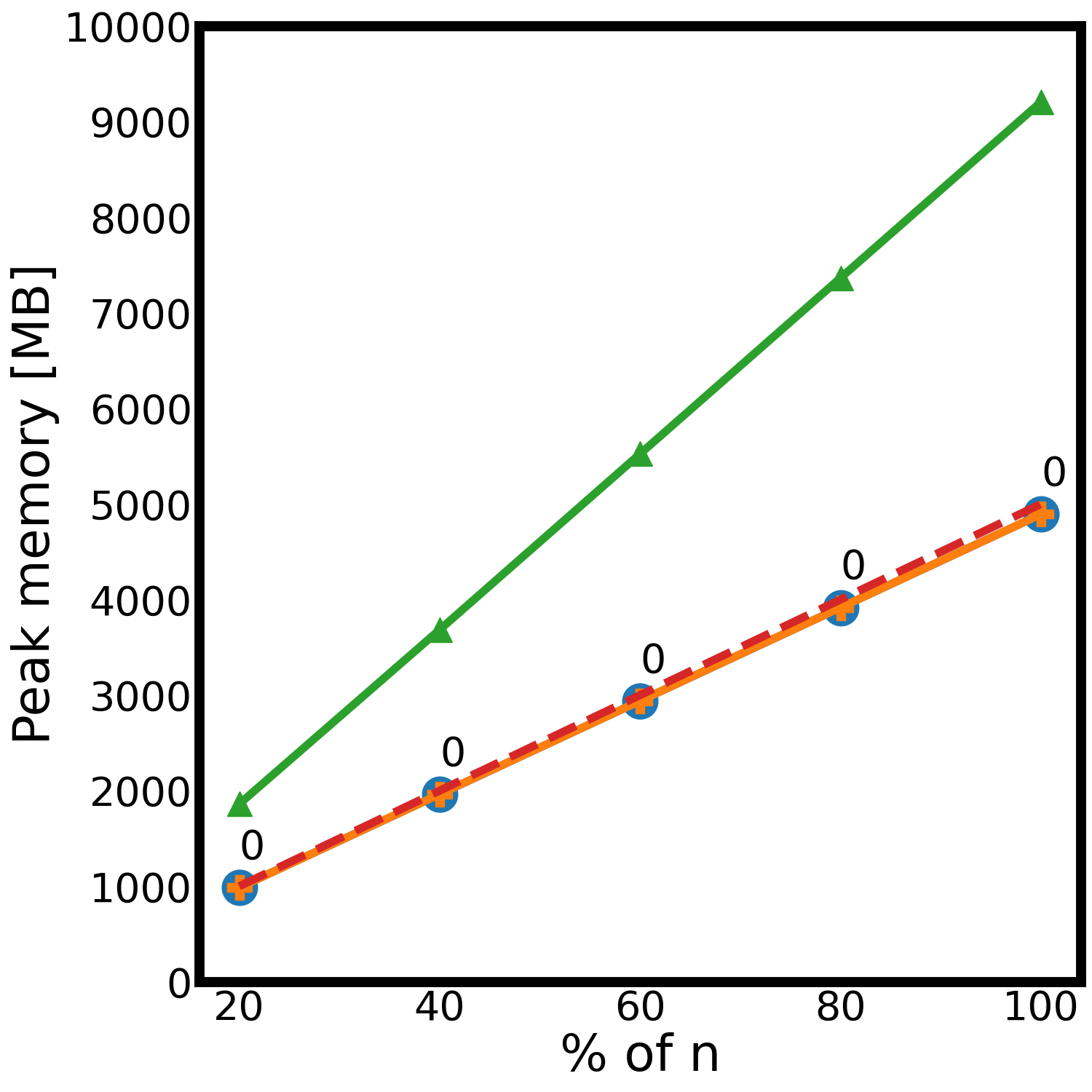}}\hspace{1mm}
        \subfloat[][\centering Memory for varying $b$ \label{mem:small:random:b}]{\includegraphics[width=0.23\textwidth]{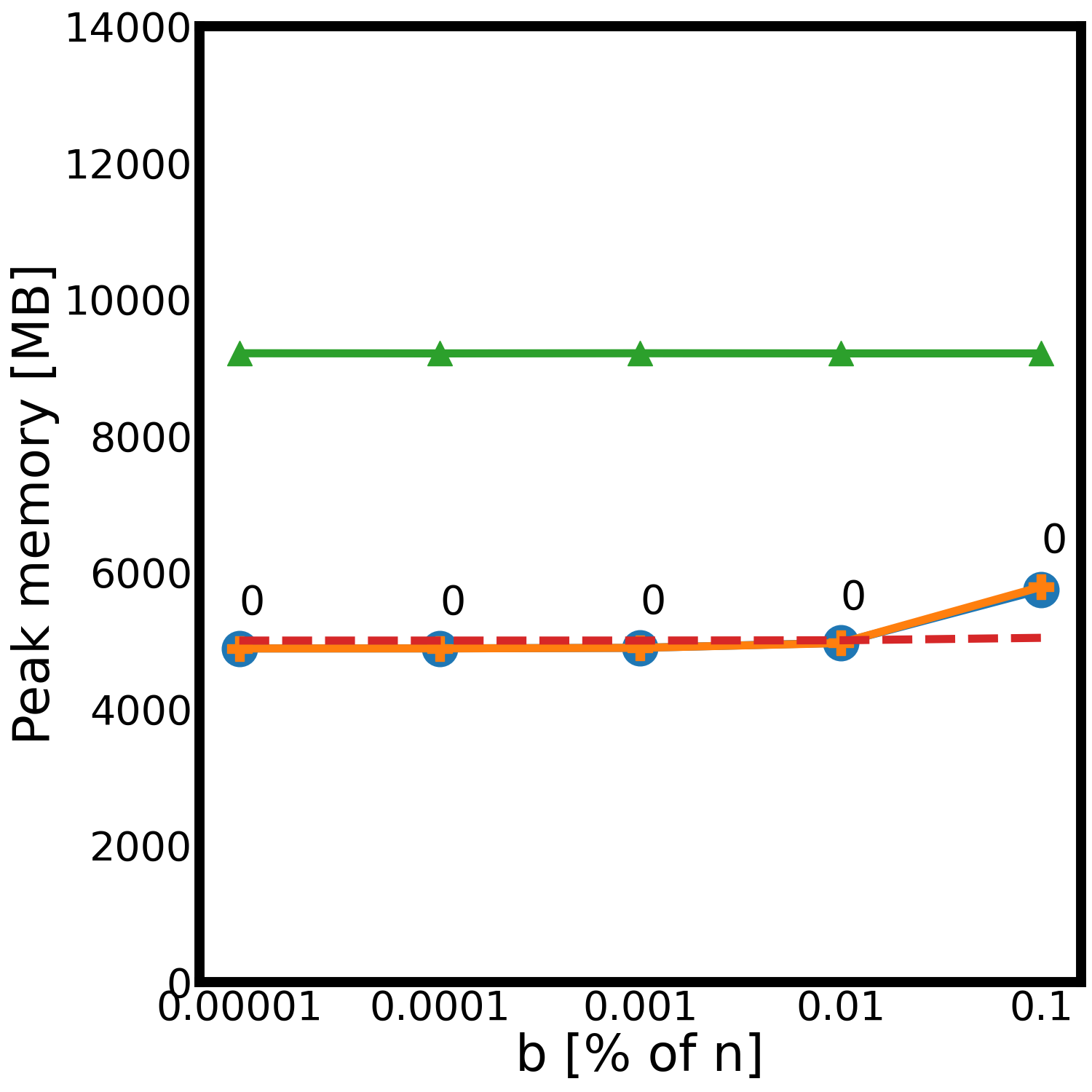}}
        \caption{\centering Results for sparse instances of \textsf{FASTQ} (top row), \textsf{AMAZON} (middle row) and \textsf{RANDOM} (bottom row). The exact value of $b'$ is on top of the points of the \textsc{PA} curves to highlight the relevance of Theorem~\ref{the:param}. 
        Notably in all but two instances we have that $b'=0$. 
        By default, we used $b=0.001\%\cdot n$ when varying $n$; and the whole dataset when varying $b$.\label{fig:small}} 
\end{figure}

\paragraph{Datasets.} 
We have used the following datasets choosing suffixes uniformly at random in all cases. For the first set of experiments (\emph{sparse instances}) we chose $b\in[n/10^7,n/10^3]$ and for the second one (\emph{dense instances}) we chose $b\in[2\cdot n/10^2,n/10]$: 
\begin{enumerate}
\item A file of $n=11,038,279,710$ bytes ($11.04$ GB), with $|\Sigma|=64$, containing $689,781$ human DNA reads of total length is $5,455,971,336$, along with their quality scores.\footnote{\url{https://github.com/nanopore-wgs-consortium/NA12878/blob/master/nanopore-human-genome/rel_3_4.md}~(FAF01132)} The reads have been sequenced using an Oxford Nanopore MinION device. We denote this dataset by \textsf{FASTQ}.
\item A file of $n=10,225,926,270$ bytes ($10.23$ GB), with $|\Sigma|=96$, containing 21,928,568 Amazon reviews (Home and Kitchen).\footnote{\url{https://cseweb.ucsd.edu/~jmcauley/datasets/amazon_v2/}(Home and Kitchen)} We denote this dataset by \textsf{AMAZON}.
\item A synthetically generated file of $n=5, 000, 000,000$ bytes ($5$GB), with $|\Sigma|=26$, containing a single string of ASCII codes. Every letter of this string was selected uniformly at random. We denote this dataset by \textsf{RANDOM}.
\end{enumerate}

\paragraph{Setup.} 
The experiments ran on an Intel Xeon Gold 648 CPU at 2.5GHz with 256GB RAM. All programs were compiled with \texttt{g++} version \texttt{10.2.0} at the \texttt{-Ofast} optimization level. 

\begin{figure}[!ht]
     \centering
        \includegraphics[width=0.5\textwidth]{figures/legend.png}
        
        \subfloat[][\centering Time for varying $n$\label{time:large:fastq:n}] 
        {\includegraphics[width=0.23\textwidth]{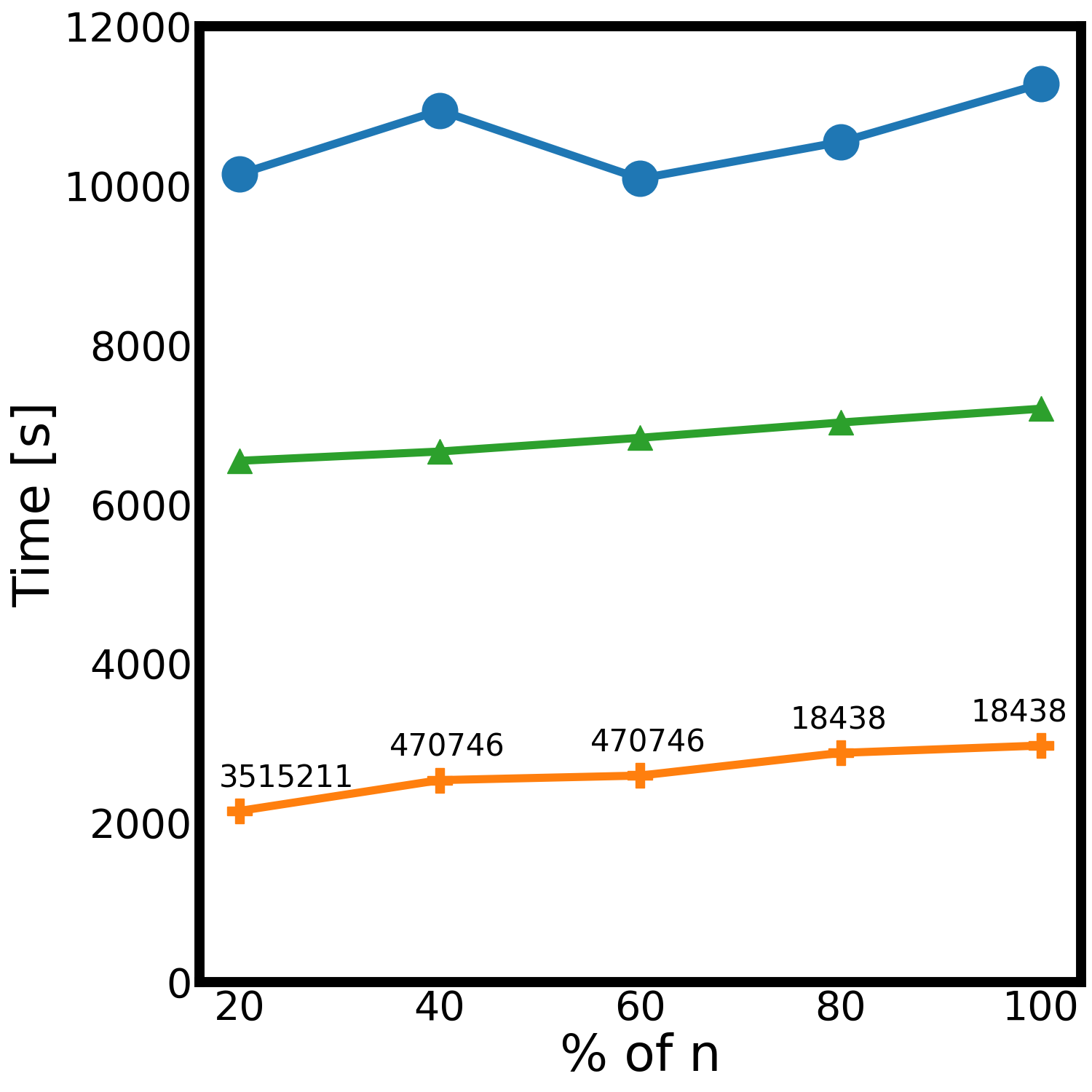}}\hspace{1mm}
        \subfloat[][\centering Time for varying $b$\label{time:large:fastq:b}]{\includegraphics[width=0.23\textwidth]{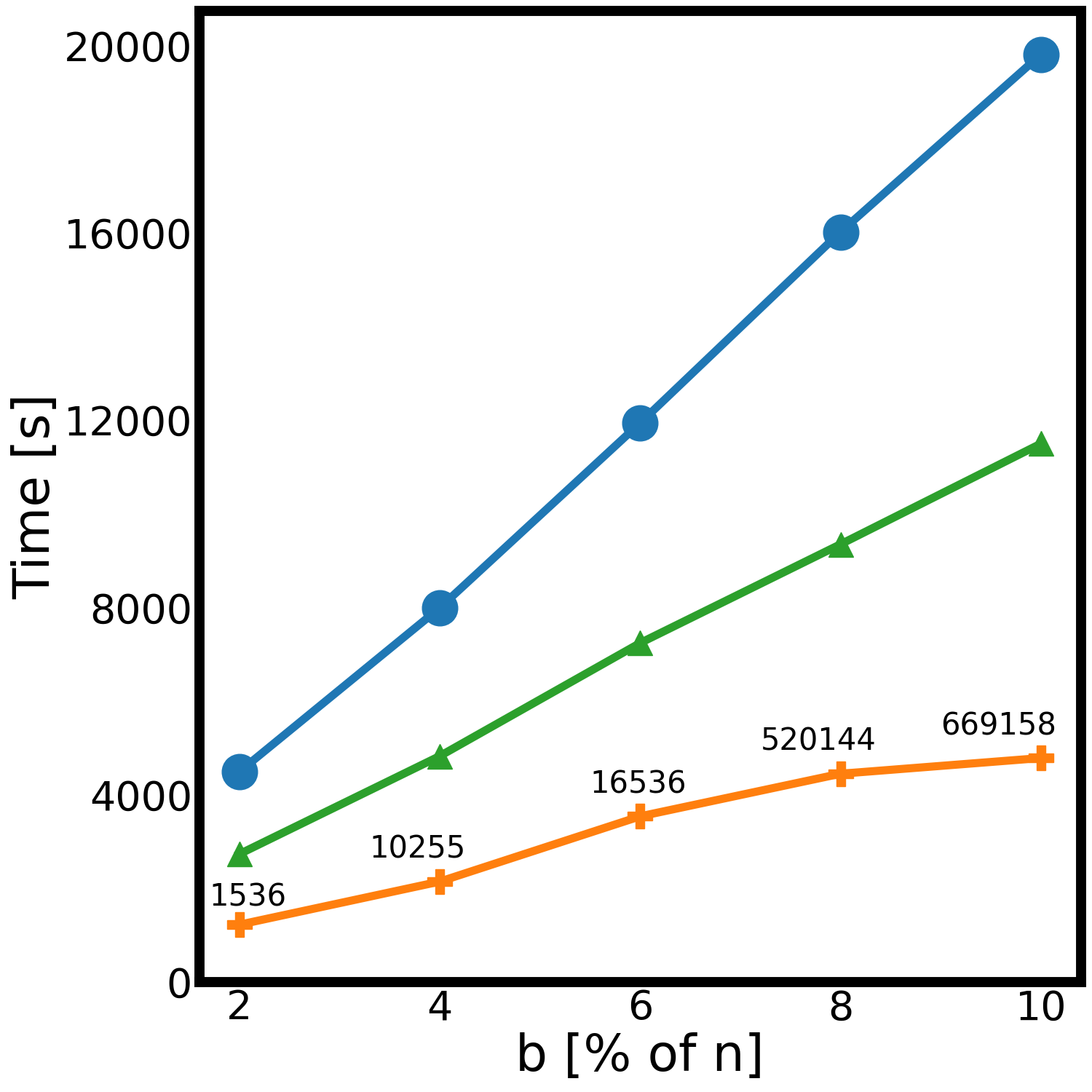}}\hspace{1mm}
        \subfloat[][\centering Memory for varying $n$ \label{mem:large:fastq:n}]{\includegraphics[width=0.23\textwidth]{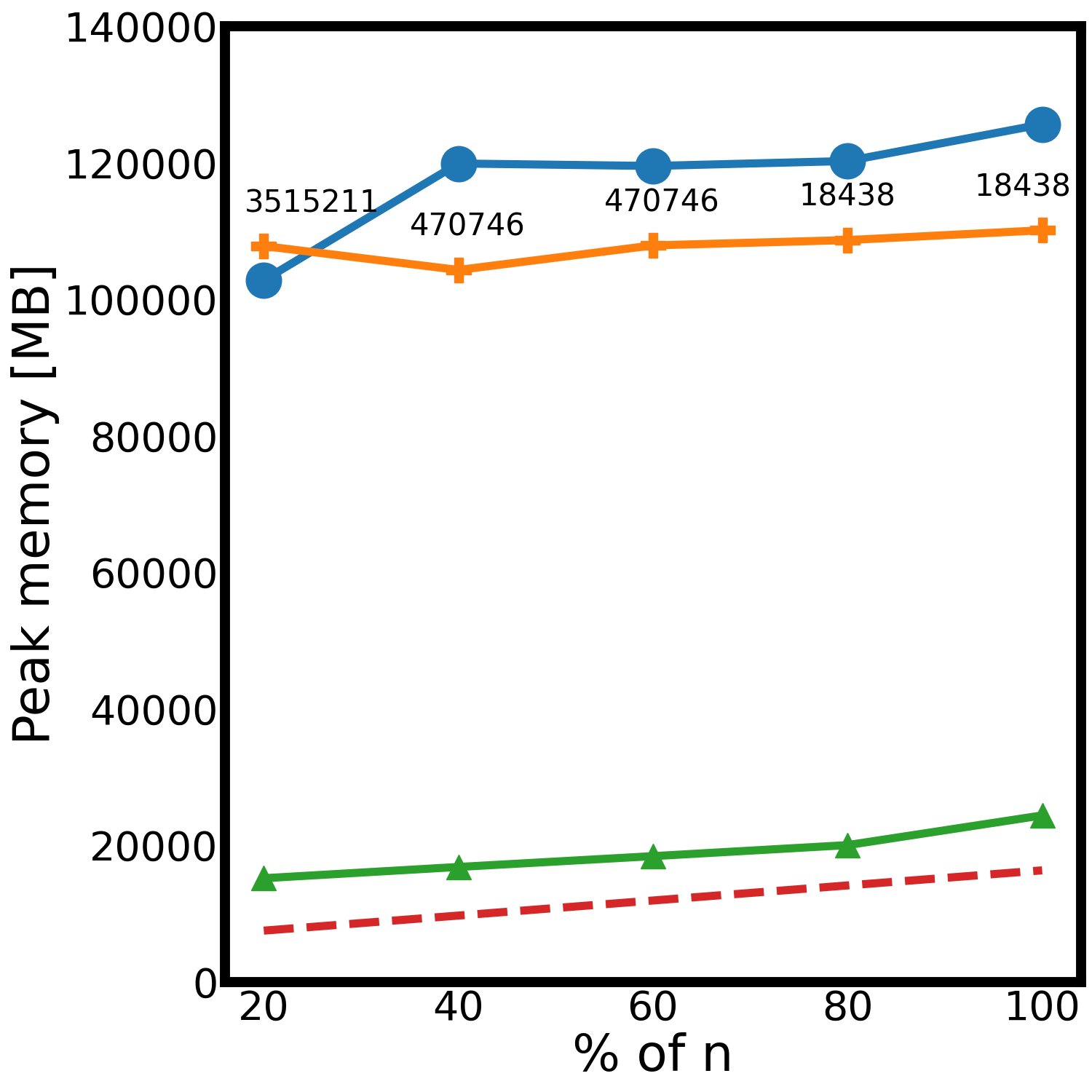}}\hspace{1mm}
        \subfloat[][\centering Memory for varying $b$ \label{mem:large:fastq:b}]{\includegraphics[width=0.23\textwidth]{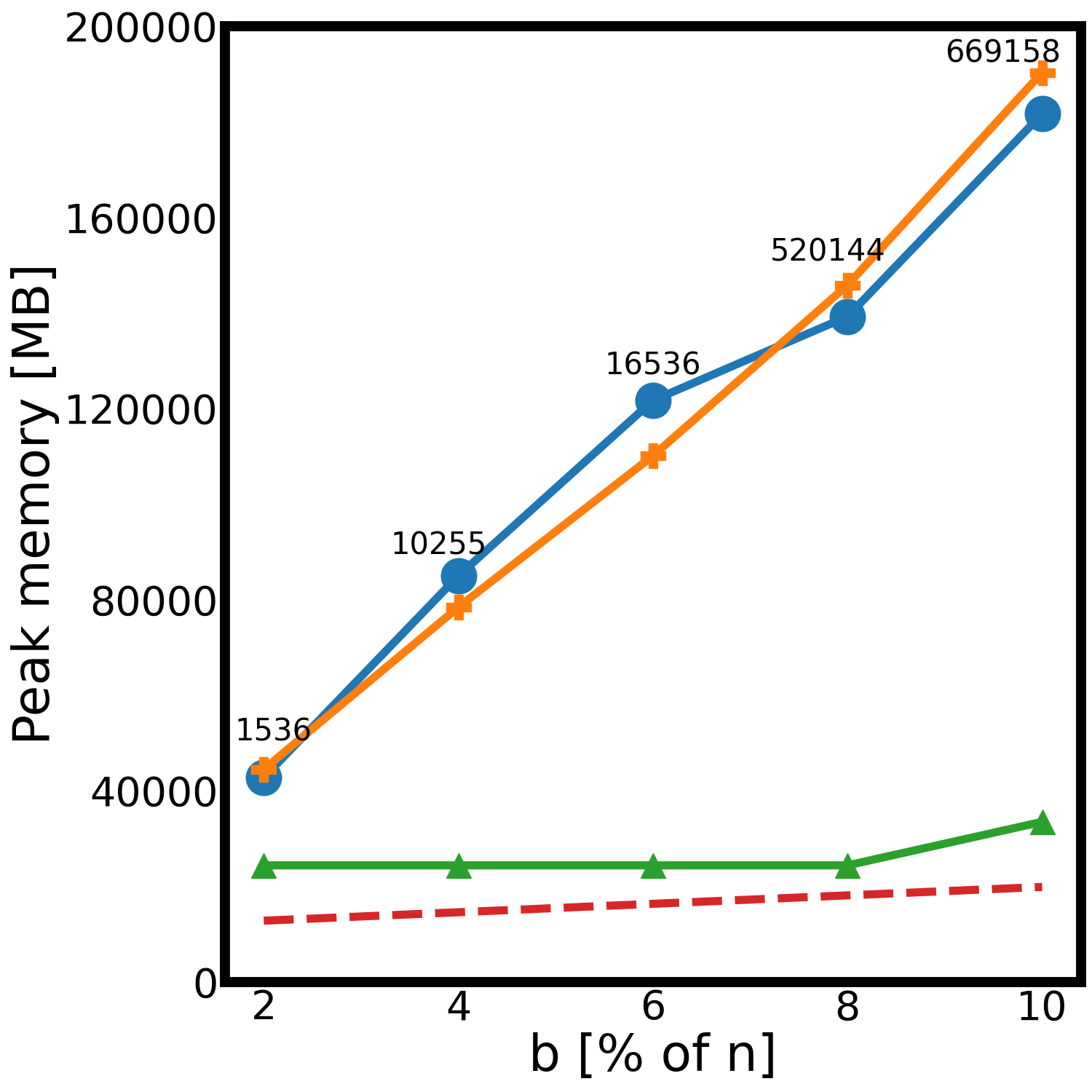}}

     \centering 
        \subfloat[][\centering Time for varying $n$\label{time:large:amazon:n}]{\includegraphics[width=0.23\textwidth]{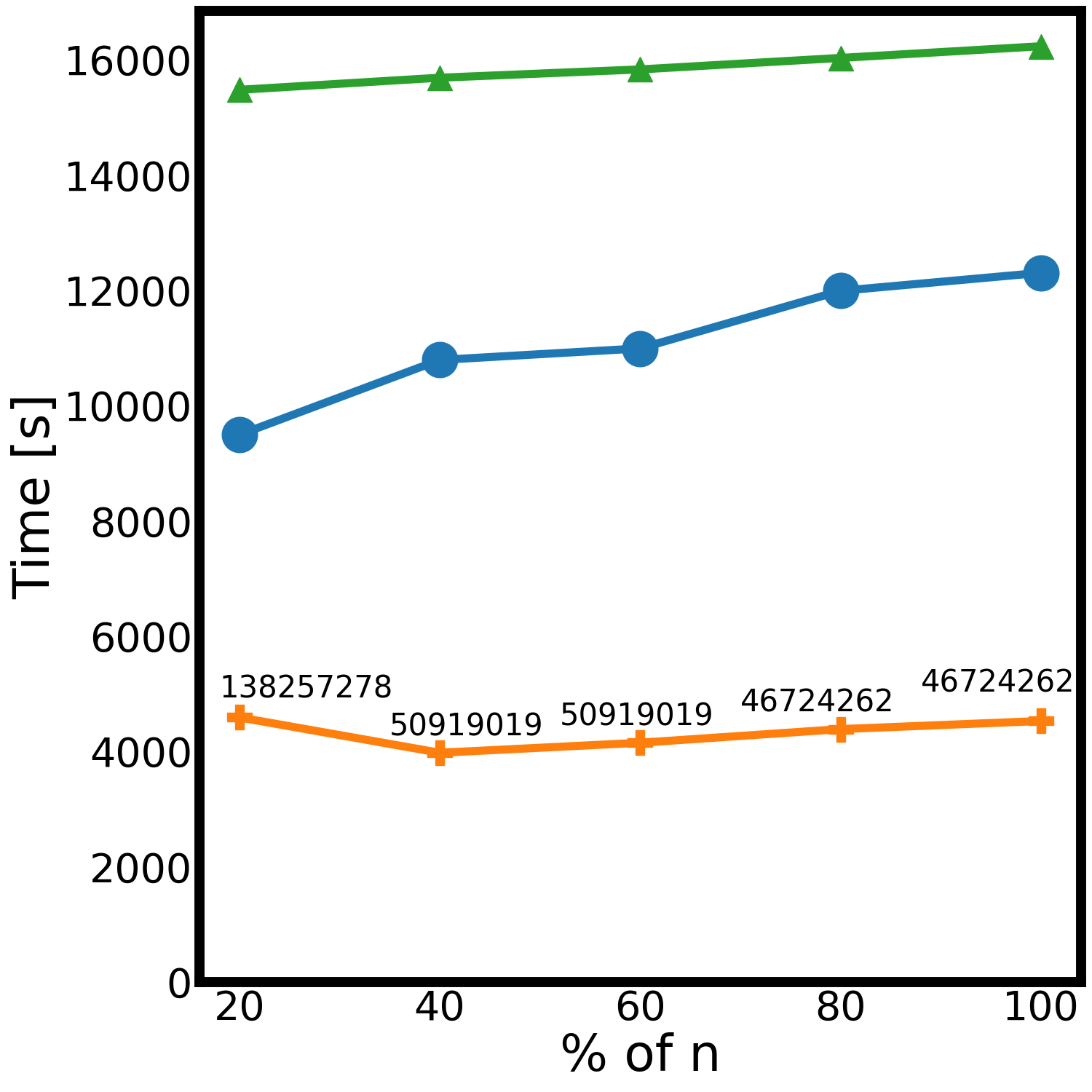}}\hspace{1mm}
        \subfloat[][\centering Time for varying $b$\label{time:large:amazon:b}]{\includegraphics[width=0.23\textwidth]{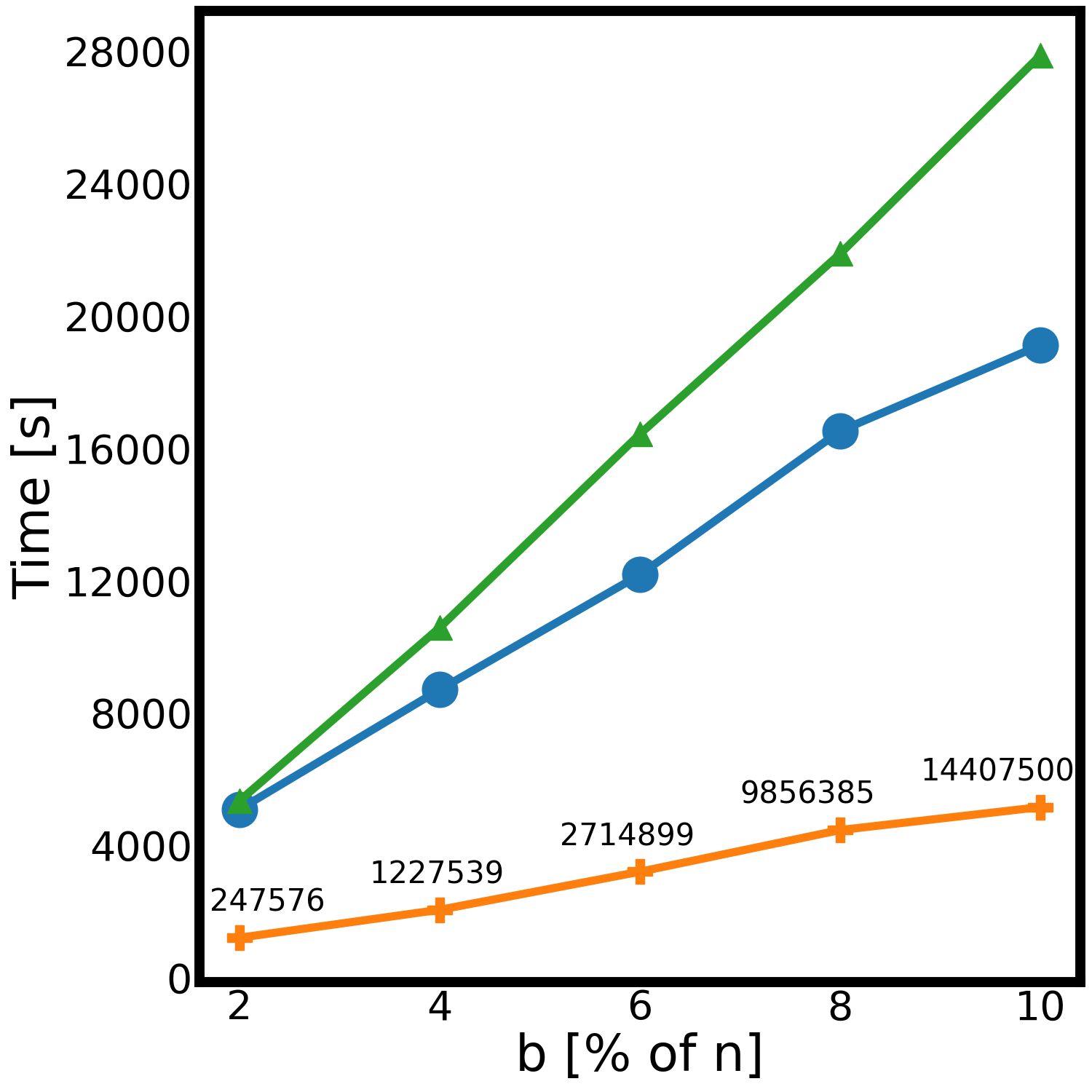}}\hspace{1mm}
        \subfloat[][\centering Memory for varying $n$ \label{mem:large:amazon:n}]{\includegraphics[width=0.23\textwidth]{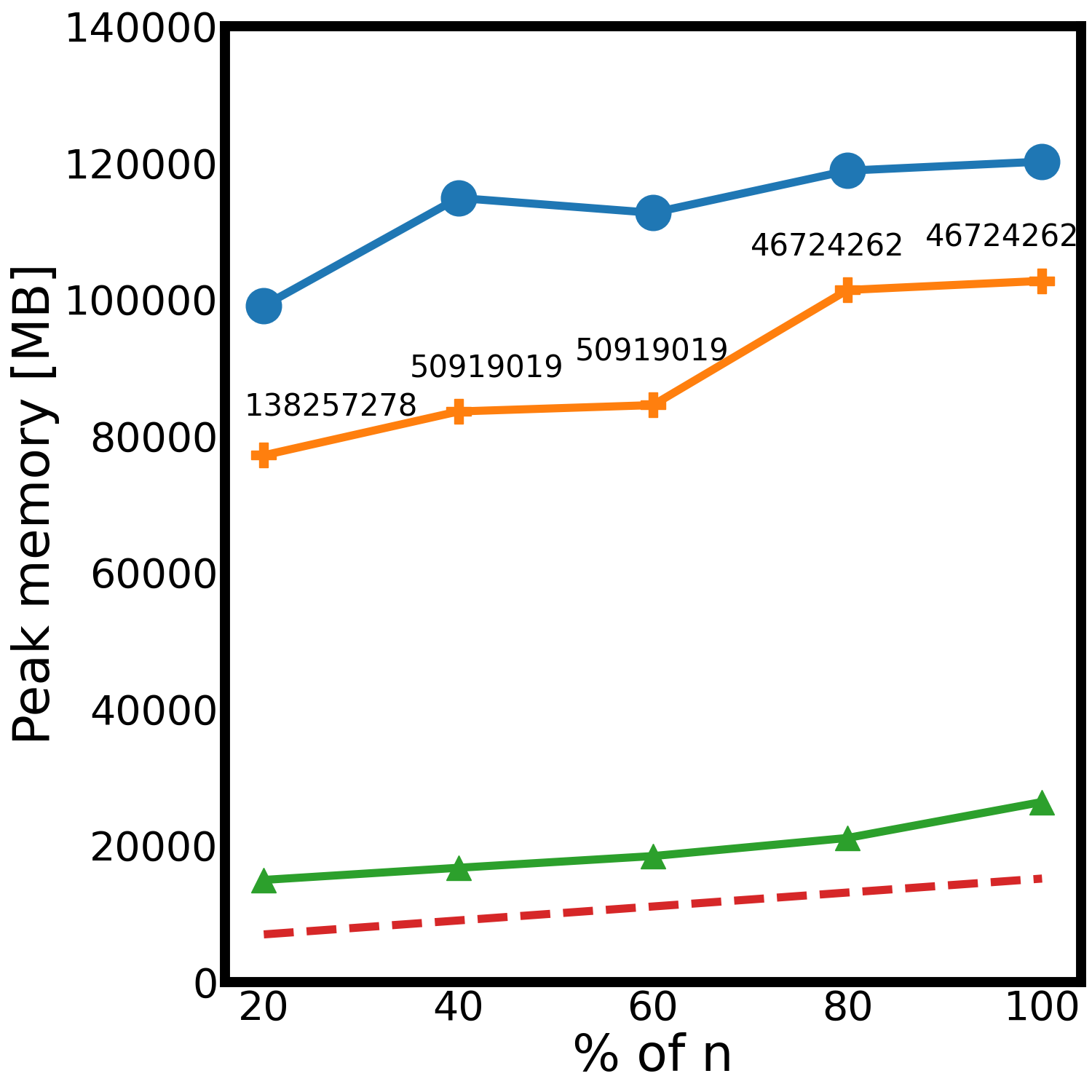}}\hspace{1mm}
        \subfloat[][\centering Memory for varying $b$ \label{mem:large:amazon:b}]{\includegraphics[width=0.23\textwidth]{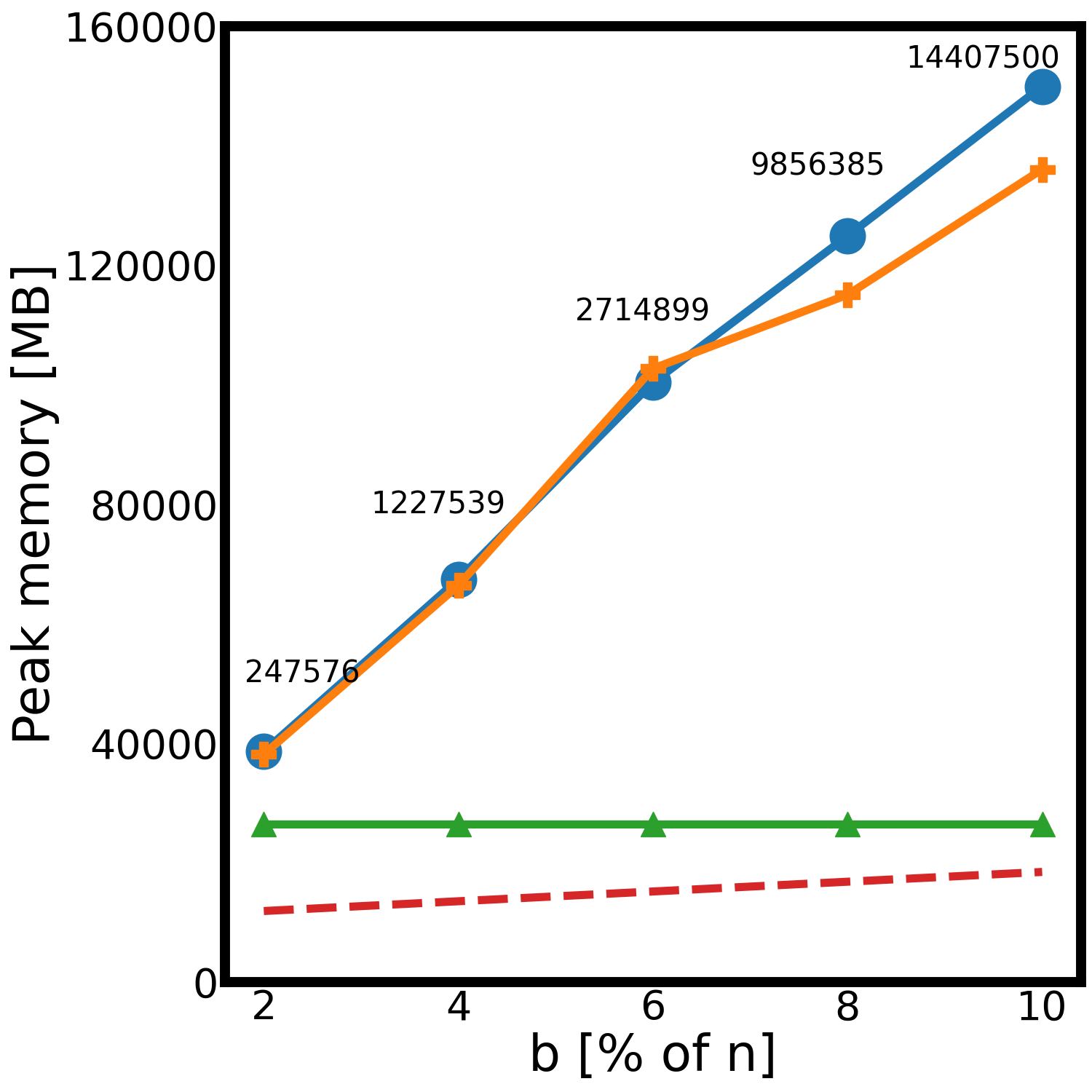}}

      \centering 
        \subfloat[][\centering Time for varying $n$\label{time:large:random:n}]{\includegraphics[width=0.23\textwidth]{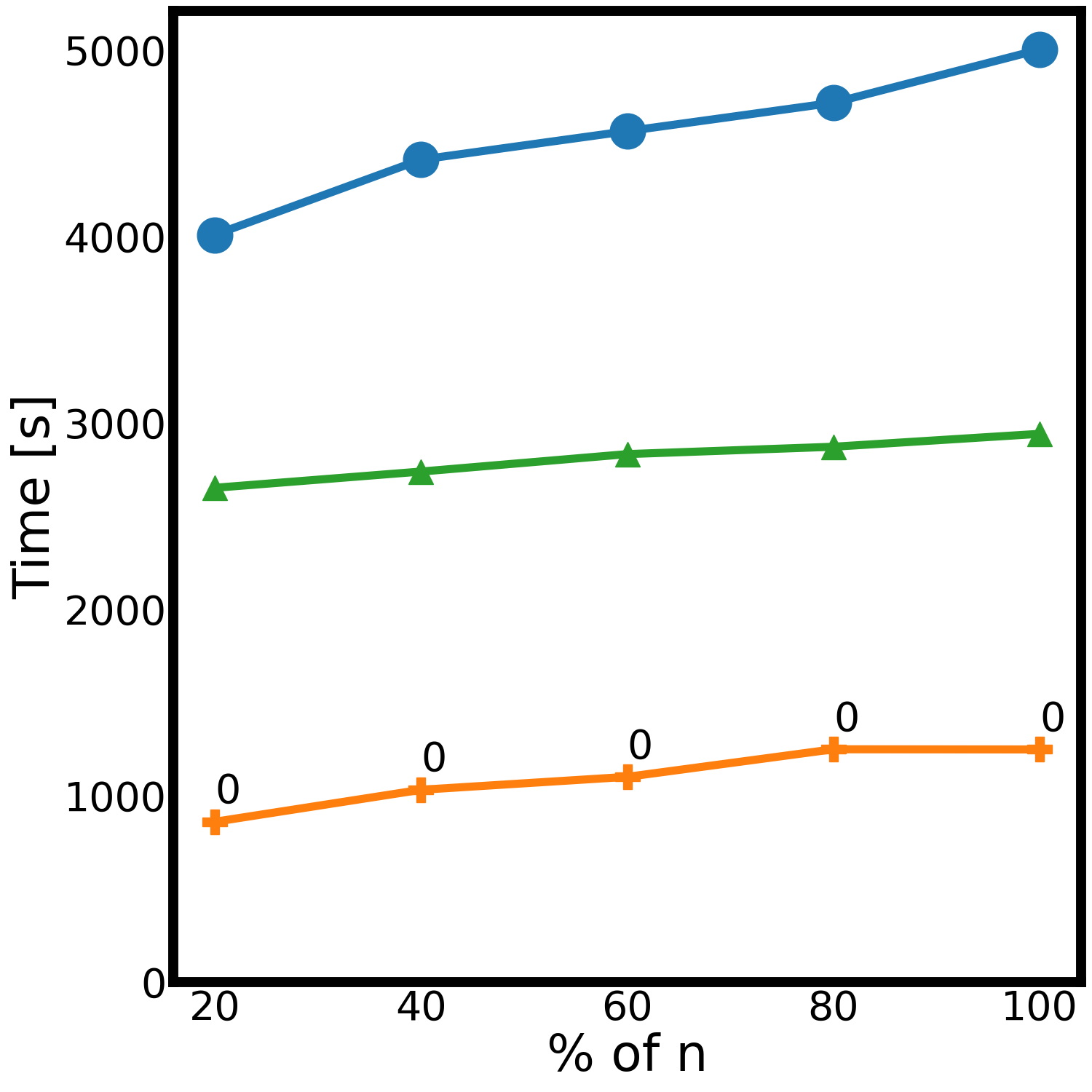}}\hspace{1mm}
        \subfloat[][\centering Time for varying $b$\label{time:large:random:b}]{\includegraphics[width=0.23\textwidth]{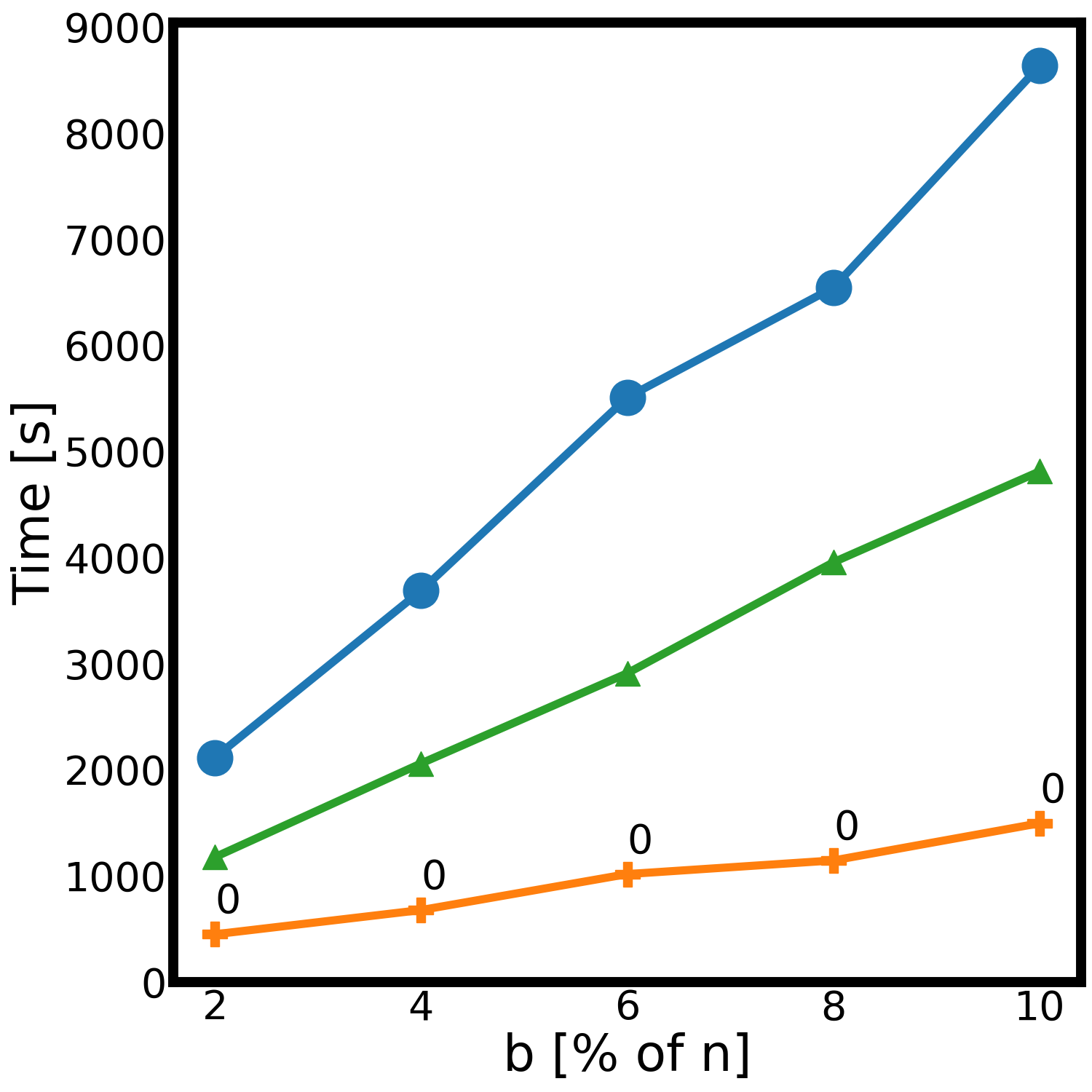}}\hspace{1mm}
        \subfloat[][\centering Memory for varying $n$ \label{mem:large:random:n}]{\includegraphics[width=0.23\textwidth]{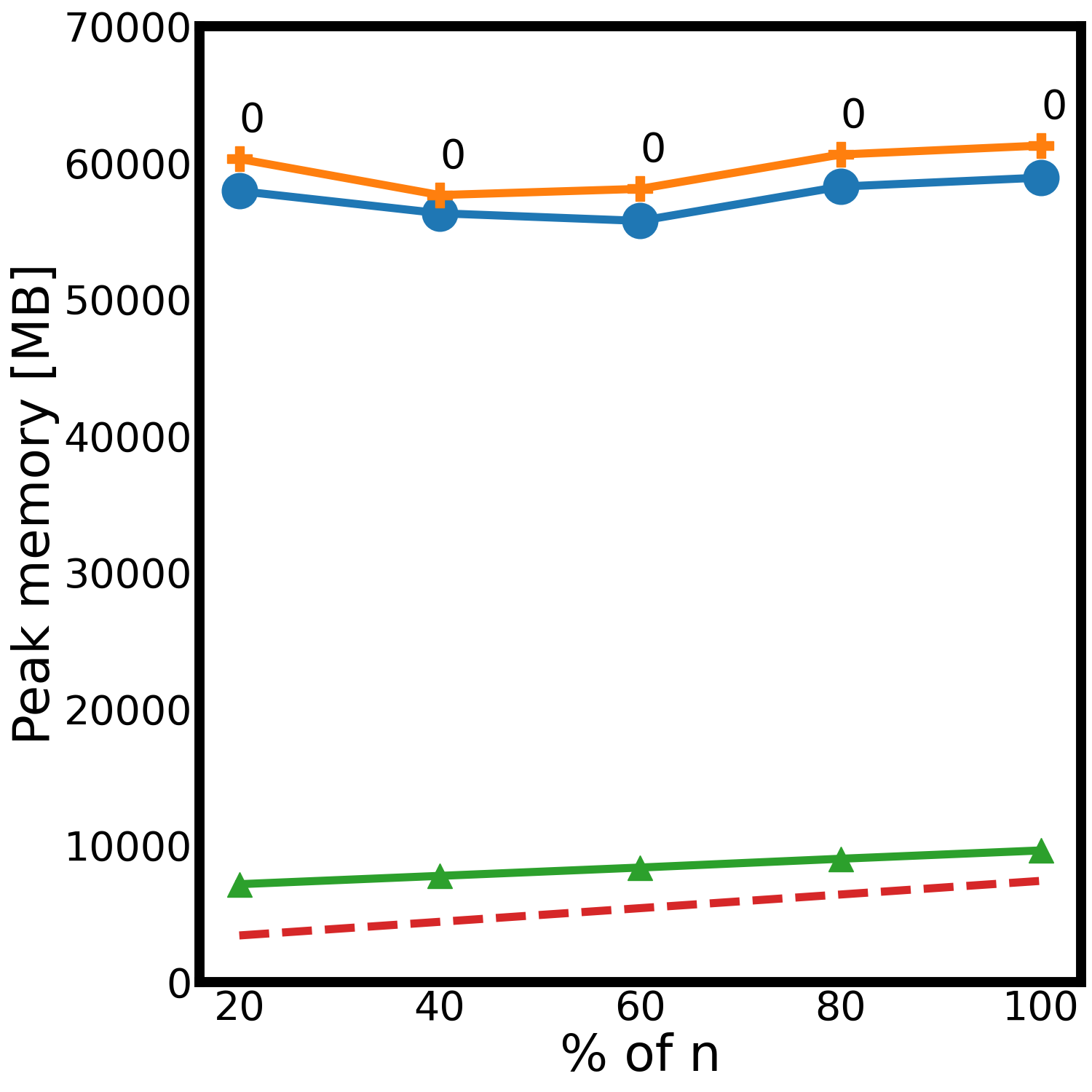}}\hspace{1mm}
        \subfloat[][\centering Memory for varying $b$ \label{mem:large:random:b}]{\includegraphics[width=0.23\textwidth]{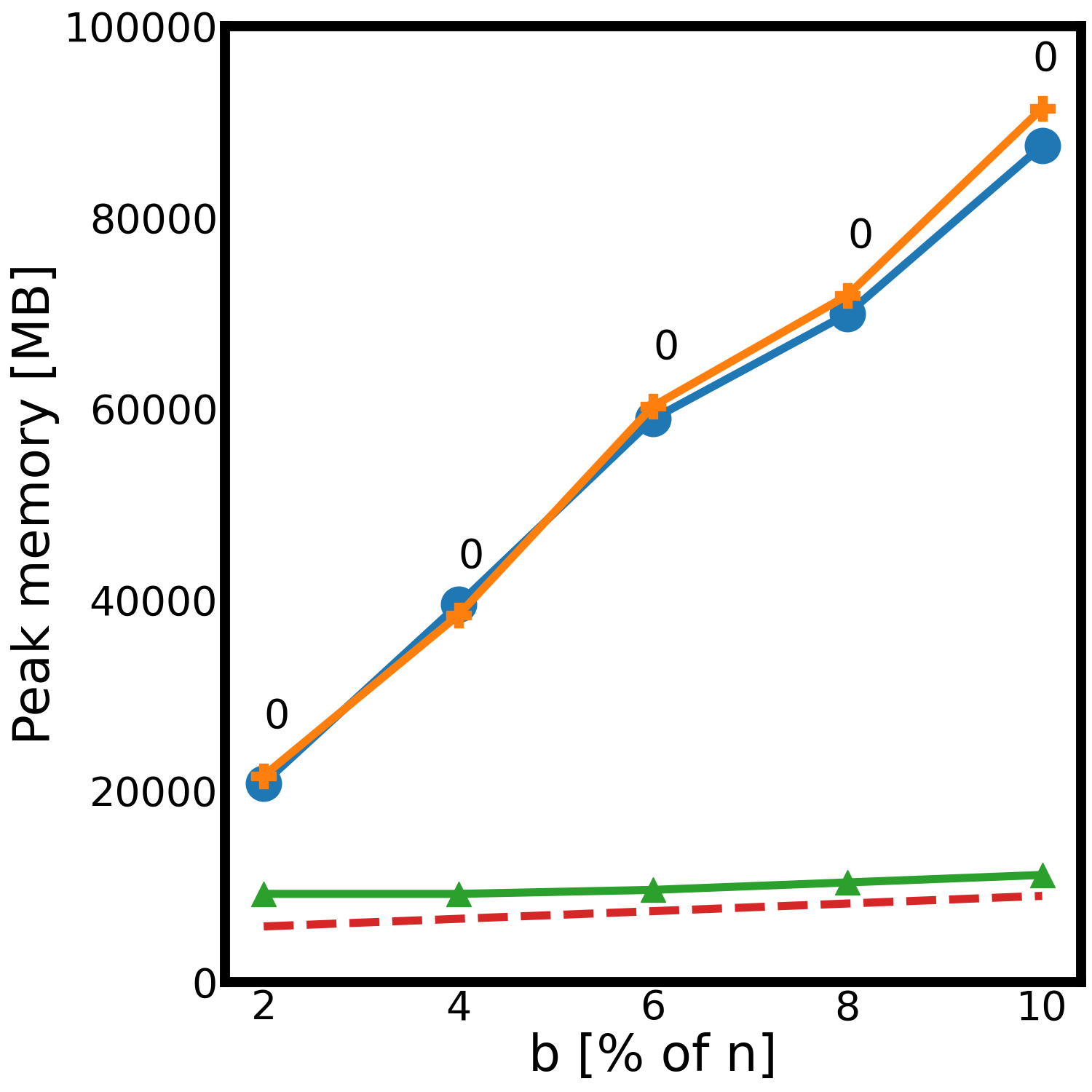}}
        \caption{\centering Results for dense instances of \textsf{FASTQ} (top row), \textsf{AMAZON} (middle row) and \textsf{RANDOM} (bottom row). The exact value of $b'$ is on top of the points of the \textsc{PA} curves to highlight the relevance of Theorem~\ref{the:param}. 
        As expected, $b'$ decreases with increasing $n$
        and increases with increasing $b$.
        By default, we used $b=6\%\cdot n$ when varying $n$; and the whole dataset when varying $b$.\label{fig:large}}
\end{figure}

\paragraph{Results (Sparse Instances).} Figure~\ref{fig:small} shows the results for the sparse instances; i.e., those where $b\in[n/10^7,n/10^3]$ is relatively small. 
By default, we used $b=0.001\%\cdot n$ for the experiments that take prefixes of the whole datasets to vary the values of $n$; and we used the whole datasets for the experiments that use varying values of $b$. Figure~\ref{time:small:fastq:n} shows that: (i) \textsc{PA} is up to $2$ times faster than \textsc{SSA-LCE} and over $2$ times faster than \textsc{MA}; and (ii) the runtime of all implementations scales linearly with $n$. Figure~\ref{time:small:fastq:b} shows that: (i) \textsc{PA} is again much faster than the other two implementations; and (ii) \textsc{PA} and \textsc{SSA-LCE} are not substantially affected (if at all) by increasing $b$ as opposed to \textsc{MA}. In conjunction, the results of Figures~\ref{time:small:fastq:n} and~\ref{time:small:fastq:b}, show that \textsc{PA} runs in linear time because $b'=0$ in all instances (see Theorem~\ref{the:param}). Figure~\ref{mem:small:fastq:n} shows that: (i) \textsc{PA} and \textsc{MA} consume about half the memory consumed by \textsc{SSA-LCE}; and (ii) the peak memory usage for all implementations scales linearly with $n$. Figure~\ref{mem:small:fastq:b} shows that: (i) \textsc{PA} and \textsc{MA} are substantially more space-efficient than \textsc{SSA-LCE}; and (ii) \textsc{PA} and \textsc{MA} start off with the space required by the input and increase the peak memory usage only slightly when $b>0.01$\% of $n$. The former is explained by the fact that \textsc{SSA-LCE} is a simplified implementation of the algorithm in~\cite{DBLP:journals/talg/Prezza21} (recall that no other implementation is available), which uses more than $\cO(1)$ words of space to sort the $b$ suffixes. The results in Figures~\ref{time:small:amazon:n}-\ref{mem:small:amazon:b} for \textsf{AMAZON} and in Figures~\ref{time:small:random:n}-\ref{mem:small:random:b} for \textsf{RANDOM} are analogous.

The benefits of \textsc{PA} in terms of time compared to \textsc{MA} and \textsc{SSA-LCE} on these sparse instances are significant.\footnote{Arguably, these sparse instances are the interesting instances for \emph{sparse} suffix sorting.} These benefits are explained by the fact that $b'=0$ in all but two instances; even in the instance with the largest $b'>0$, we still have that  $b'=518<\lfloor b/\log b \rfloor=439, 149$ satisfying the condition in Theorem~\ref{the:param}. These results \emph{fully justify} the design goals of \textsc{PA}. Notably, in these sparse instances, \textsc{SSA-LCE} consumes more memory than \textsc{PA} and \textsc{MA}.

\paragraph{Results (Dense Instances).} Figure~\ref{fig:large} shows the results for the dense instances, where $b\in[2\cdot n/10^2,n/10]$. In these instances, $b$ is relatively large (from $2$\% to $10$\% of $n$). By default, we used $b=6\%\cdot n$ for the experiments that take prefixes of the whole datasets to vary the values of $n$; and we used the whole datasets for the experiments that use varying values of $b$.
Figure~\ref{time:large:fastq:n} shows that: (i) \textsc{PA} is up to $2$ times faster than \textsc{SSA-LCE} and over $3$ times faster than \textsc{MA}; and (ii) the runtime of all implementations increases only slightly with $n$.
Figure~\ref{time:large:fastq:b} shows that: (i) \textsc{PA} is again many times faster than the other two implementations; and (ii) \textsc{PA} is not as substantially affected by increasing $b$ as \textsc{SSA-LCE} and \textsc{MA} are. 
In conjunction, the results of Figures~\ref{time:large:fastq:n}-\ref{time:large:fastq:b} show that, in these dense instances, increasing $b$ becomes considerably more costly than increasing $n$ (unlike in the sparse instances), due to the $\mathcal{\tilde{O}}(b)$ term in the time complexity of all algorithms being much more costly than the $\mathcal{O}(n)$ term.
Figure~\ref{mem:large:fastq:n} shows that: (i) \textsc{PA} and \textsc{MA} consume about $5$ times more memory compared to \textsc{SSA-LCE} due to the large values of $b$ (as expected by the space anlysis of the algorithms); and (ii) the peak memory usage for all implementations increases only slightly with $n$ (as expected due to the larger input string). Figure~\ref{mem:large:fastq:b} shows that: (i) the peak memory usage of \textsc{PA} and \textsc{MA} scales linearly with $b$; and (ii) \textsc{PA} and \textsc{MA} consume up to $5$ times more memory in comparison to \textsc{SSA-LCE}. Figures~\ref{time:large:amazon:n}-\ref{time:large:amazon:b} show that for \textsf{AMAZON} \emph{both} \textsc{PA} and \textsc{MA} are faster than \textsc{SSA-LCE} with \textsc{PA} being over $2$ times faster than \textsc{SSA-LCE}. 
This happens because \textsc{SSA-LCE} becomes considerably slower when $b'$ is large; i.e., when many pairs of suffixes have long LCP values. In such a case, comparing any two suffixes takes $\cO(\log n)$ time for \textsc{SSA-LCE} instead of $\cO(1)$ time, and so the $\mathcal{\tilde{O}}(b)$ term becomes a bottleneck.
The peak memory usage results in Figures~\ref{mem:large:amazon:n}-\ref{mem:large:amazon:b} for \textsf{AMAZON} and Figures~\ref{time:large:random:n}-\ref{mem:large:random:b} for \textsf{RANDOM} are analogous to those for \textsf{FASTQ}.

Notably, the benefits of \textsc{PA} compared to \textsc{MA} and \textsc{SSA-LCE} in terms of time remain significant even in these harder, dense instances, where the values of $b'$ increase compared to the $b'$ values in the analogous sparse instances (see Figure~\ref{fig:small}). 
Consider the instance for $b=10\% \cdot n$ on \textsf{AMAZON} (Figure~\ref{time:large:amazon:b}).  
We still have that $b'=14407500<\lfloor b/\log b \rfloor=34166616$.
In these dense instances, \textsc{SSA-LCE} consumes much less memory than \textsc{MA} and \textsc{PA}; this is expected as \textsc{SSA-LCE} works in the restore model: it overwrites (and thus uses the space occupied by) $T$.

\section*{Acknowledgments}
This work is partially supported by the PANGAIA and ALPACA projects that have received funding from the European Union's Horizon 2020 research and innovation programme under the Marie Skłodowska-Curie grant agreements No 872539 and 956229, respectively. Hilde Verbeek is supported by a Constance van Eeden Fellowship.

\bibliographystyle{plain}
\bibliography{references}

\end{document}